\documentclass[
]{ceurart}

\usepackage{listings}
\usepackage{graphbox}
\DeclareRobustCommand{\DLLogo}{%
  \begingroup\normalfont
  \kern-1.75pt\includegraphics[align=c,height=1.25\baselineskip]{dl-logo}\kern-1.5pt%
  \endgroup
}

\usepackage[mathlines]{lineno}
\usepackage{etoolbox}
\usepackage{url}
\usepackage[obeyFinal]{todonotes}
\usepackage[utf8]{inputenc}
\usepackage{csquotes}
\usepackage{rotating}
\usepackage{graphicx}
\usepackage{multirow}
\usepackage{multicol}
\usepackage{tikz}
\usetikzlibrary{arrows}
\usepackage{macros}

\usepackage{comment}

\newif\ifshortversion
\shortversiontrue

\ifshortversion
  \includecomment{shortonly}
  \excludecomment{longonly}
\else
  \includecomment{longonly}
  \excludecomment{shortonly}
\fi

\usepackage{amsthm}
\usepackage[capitalise]{cleveref}
\crefformat{footnote}{#2\footnotemark[#1]#3}
\newtheorem{theorem}{Theorem}
\newtheorem{lemma}[theorem]{Lemma}
\newtheorem{corollary}[theorem]{Corollary}
\newtheorem{observation}[theorem]{Observation}

\newtheorem{claim}{Claim}
\theoremstyle{definition}
\newtheorem{definition}[theorem]{Definition}

\newtheorem{example}[theorem]{Example}

\AddToHook{env/lemma/begin}{\crefalias{theorem}{lemma}}
\AddToHook{env/corollary/begin}{\crefalias{theorem}{corollary}}
\AddToHook{env/observation/begin}{\crefalias{theorem}{observation}}
\AddToHook{env/fact/begin}{\crefalias{theorem}{fact}}
\AddToHook{env/proposition/begin}{\crefalias{theorem}{proposition}}
\AddToHook{env/claim/begin}{\crefalias{theorem}{claim}}
\AddToHook{env/definition/begin}{\crefalias{theorem}{definition}}
\AddToHook{env/remark/begin}{\crefalias{theorem}{remark}}
\AddToHook{env/example/begin}{\crefalias{theorem}{example}}

\usepackage[obeyFinal]{todonotes}
\usepackage{hyperref}

\AtBeginDocument{%
  \makeatletter
  \gdef\xmp@Title{Repair Semantics for EL-bot}%
  \makeatother
}
\begin{document}



\title{
The More the Merrier: Combining Properties for ABox Abduction under
Repair Semantics in ELbot
}[The More the Merrier: Combining Properties for ABox Abduction under Repair Semantics in $\mathcal{EL}_\bot$]


\author[1]{Anselm Haak}[%
orcid=0000-0003-1031-5922,
email=anselm.haak@uni-paderborn.de,
]
\address[1]{Knowledge Representation Group, Paderborn University, Germany}

\author[2]{Patrick Koopmann}[%
orcid=0000-0001-5999-2583,
email=p.k.koopmann@vu.nl,
url=https://pkoopmann.github.io,
]
\address[2]{Knowledge in Artificial Intelligence, 
Vrije Universiteit Amsterdam, 
The Netherlands}

\author[3]{Yasir Mahmood}[%
orcid=0000-0002-5651-5391,
email=yasir.mahmood@uni-paderborn.de,
]
\address[3]{Data Science Group, Paderborn University, Germany}

\author[1]{Anni-Yasmin Turhan}[%
orcid=0000-0001-6336-335X,
email=turhan@uni-paderborn.de
]

\begin{abstract}
  Abduction is a central approach to explain missing entailments from a knowledge base by providing a hypothesis, that would, if added to the knowledge base, make the missing entailment become true.
  Abduction under repair semantics has recently been investigated in detail, where several desirable properties and optimality criteria were considered, such as signature-restrictions and minimality in size and of introduced conflicts. Naturally, hypotheses that satisfy more than one of these properties or combine a property with an optimality criterion would be even more desirable for applications.
  So far, such hypotheses have not been investigated in the literature. 
  In the present paper, we consider the ABox abduction problem for hypotheses satisfying more than one property or additional optimality criteria, for \ELbot under brave and AR semantics. Our main observation is that often requiring additional properties for hypotheses does not lead to an increase of complexity.
\end{abstract}


\maketitle

\section{Introduction}

Explainability remains a core feature of knowledge-based AI systems, including
those using description logic ontologies. Developing methods that can also
practically explain decisions made by such systems is crucial in order to
exploit this feature. While there is a lot of work on
explaining reasoning results under classical
semantics~\cite{SchlobachC03,AlrabbaaBFHKKKK24,Koopmann25}, this problem is less
investigated for non-classical semantics. In this paper, we consider
\emph{repair-based semantics}, which are a family of well-understood semantics that allow
to define meaningful entailments from inconsistent knowledge
bases~\cite{BienvenuB16}. As the name suggests, these semantics are based on
\emph{repairs}, which are maximal consistent subsets of the ABox. Depending on
the chosen
semantics, we may then consider those entailments that hold in at least one of
the repairs (\emph{brave} semantics), or only those entailments that hold in
all of the repairs (\emph{AR} semantics).

Since inconsistencies can easily
occur in realistic data, inconsistency-tolerant semantics are
often indispensable. Unfortunately, they also make it more
challenging to understand entailments from a user perspective. To explain why an
assertion is entailed,
we may need to provide explanations for the different repairs at the same time.
For explaining \emph{missing} entailments, the situation is even more complex:
Indeed, under repair-based semantics, two things can contribute to the
non-entailment of an assertion $\alpha$:
1)~some other assertions \enquote{block} the entailment
of $\alpha$, because they create a conflict with required assertions, or
2)~assertions that are needed for the entailment are missing. While
explanations based on Situation~1 have been investigated in detail~\cite{bienvenu2019computing,lukasiewicz2022explanations}, we focus here
on explanations based on Situation~2.
Interestingly, no one so far considered explanations that combine Situations~1) and~2), i.e., where assertions can both be added and removed.
We prefer to first fully understand abductive reasoning under repair-based semantics, and investigate more advanced notions of explanations in
the future.
Focusing on Situation~2, this corresponds to explaining missing
entailments using abductive reasoning, which is also a typical means of
explaining missing entailments under classical
semantics~\cite{Elsenbroich2006,Klarman2011,Calvanese2013,du2015towards,%
AlrabbaaBFHKKKK24,Ceylan2020,Koopmann21a}.
In particular, given a knowledge base $\Kmc$ and
an assertion $\alpha$, \emph{ABox abduction} is concerned with finding a set of
assertions $\Hmc$, called \emph{hypothesis}, s.t. $\Kmc\cup\Hmc\models\alpha$,
where usually additional properties are required for $\Hmc$. ABox abduction is
not only useful for explaining missing entailments, it can also be used for
diagnosis~\cite{Elsenbroich2006,ObeidOMO19,Ceylan2020,Koopmann21a} and for
repairing incomplete knowledge
bases~\cite{WeiKleinerDragisicLambrix2014,du2017practical,Haifani2022}.
If we replace the classical entailment relation with entailment under
repair-based entailment, we obtain the notion of ABox abduction under
repair-based semantics.

\newcommand{\dl}[1]{\textit{#1}}

\newcommand{\exA}{\dl{High}}
\newcommand{\exB}{\dl{Low}}
\newcommand{\exC}{\dl{GlycemicCrisis}}
\newcommand{\exs}{\dl{glucoseLevel}}
\newcommand{\exr}{\dl{experiences}}
\newcommand{\exE}{\dl{OverdosedInsulin}}
\newcommand{\exF}{\dl{Ketoacidosis}}
\newcommand{\exO}{\dl{DiabeticComa}}
\newcommand{\exa}{\dl{patient}}
\newcommand{\exb}{\dl{l}}
\newcommand{\exc}{c}
\begin{example}
 Consider the following excerpt from a medical ontology $\Tmc$ on diabetes:
 \begin{align*}
  &\exA\sqcap\exB\sqsubseteq \bot \qquad
  \exists\exs.\exA\sqcap\exE\sqsubseteq \exO \\
  &\exists\exs.\exA\sqsubseteq \exC \qquad\,
  \exists\exs.\exB\sqsubseteq \exC \qquad\\
  &\exC\sqcap\exF\sqsubseteq \exO
 \end{align*}
 Our ABox $\Amc$ contains information about a patient obtained through a glucose
 monitor and a finger stick sensor:
  \[
  \exs(\exa,\exb)\qquad \exA(\exb)\qquad \exB(\exb)
 \]
 Because the two sensors provided contradictory information, the knowledge base
 is inconsistent.
 We observe that the patient passed out and want to use abduction to obtain
 a possible diagnosis. Since our data is contradictory, we apply repair-based
 semantics. Both $\{\ \exE(\exa)\ \}$ and $\{\ \exF(\exa)\ \}$ are hypotheses
 under Brave semantics, while only $\{\ \exF(\exa)\ \}$ is a hypothesis under AR semantics.
 Indeed, the cautious doctor should investigate this hypothesis first.
\end{example}

\begin{figure}[h]
	\centering
	\begin{minipage}{0.35\textwidth}
		\centering
		\scalebox{0.7}{%
			\newcommand{\bravecolor}[1]{\textcolor{red}{#1}}
\newcommand{\arcolor}[1]{\textcolor{blue}{#1}}

%
\begin{tikzpicture}[font=\large]
	\tikzset{
		prop/.style= {draw, rounded corners, thin},
		brave/.style={text=red},
		ar/.style={text=blue}, 
		every path/.style={very thick},
		comb/.style={midway, draw=gray, fill=gray!30, rounded corners, inner sep=2pt}
	}
	\node[prop, label={above:\bravecolor{\SigmaP}/\arcolor{\co\NP}}]%
	(cc) at (150:6) {conflict-confining};
	\node[prop, label={below:\bravecolor{\NP}/\arcolor{\SigmaP}}]%
	(sig) at (-150:6) {$\Sigma$-restricted};
	\node[prop, label={above:\bravecolor{\NP}/\arcolor{\SigmaP}}]%
	(nt) at (30:0) {non-trivial};
%
%
	\draw (cc) to[bend right=20] node[comb] {\bravecolor{\SigmaP}/\arcolor{\SigmaP}} (nt);

	\draw (sig) to[bend left=20] node[comb] {\bravecolor{\NP}/\arcolor{\SigmaP}} (nt);

	\draw (cc) to[bend left=20] node[comb] {\bravecolor{\SigmaP}/\arcolor{\SigmaP}} (sig);

\end{tikzpicture}
%

		}
	\end{minipage}
\hspace{.5cm}
	\begin{minipage}{0.57\textwidth}
		\centering
		\scalebox{0.7}{%
			\newcommand{\bravecolor}[1]{\textcolor{red}{#1}}
\newcommand{\arcolor}[1]{\textcolor{blue}{#1}}

%
\begin{tikzpicture}[font=\large]
	\tikzset{
		prop/.style= {draw, rounded corners, thin},
		brave/.style={text=red},
		ar/.style={text=blue}, 
		every path/.style={very thick},
		comb/.style={midway, draw=gray, fill=gray!30, rounded corners, inner sep=2pt}
	}
	\node[prop, label={above:\bravecolor{\DP}/\arcolor{\co\NP}}]%
	(cc) at (150:6) {conflict-confining};
	\node[prop, label={above:\bravecolor{\DP}/\arcolor{\PiP}}]%
	(subset) at (30:6) {$\subseteq$-minimal};
	\node[prop, label={above:\bravecolor{\NP}/\arcolor{\co\NP}}]%
	(nt) at (-90:0) {non-trivial or $\Sigma$-restricted};
	\node[prop, label={below:\bravecolor{\NP}/\arcolor{\co\NP}}]%
	(card) at (-150:6) {$\leq$-minimal};
	\node[prop, label={below:\bravecolor{\PiP}/\arcolor{\co\NP}}]%
	(conf) at (-30:6) {$\subseteq_c$-minimal};
%
%
	\draw (cc) to[bend right=15] node[comb] {\bravecolor{\DP}/\arcolor{\co\NP} } (nt);

	\draw (subset) to[bend right=15] node[comb] {\bravecolor{\DP}/\arcolor{\PiP}} (nt);

	\draw (subset) to[bend right=10] node[comb] {\bravecolor{\DP}/\arcolor{in $\PiP, \co\NP$-h}} (cc);

	\draw (card) to[bend right=15] node[comb] {\bravecolor{\DP}/\arcolor{\PiP}} (nt);
	\draw (card) to[bend right=10] node[comb] {\bravecolor{\PiP}/\arcolor{\co\NP}} (cc);

	\draw (conf) to[bend right=15] node[comb] {\bravecolor{\PiP}/\arcolor{\PiP}} (nt);

\end{tikzpicture}
%

		}
	\end{minipage}
	\caption{Complexity results for Existence (Left) and Verification (Right) problem under combination of properties. Each node contains a property and its complexity for brave (red) and AR (blue) semantics, highlighting our earlier results~\cite{HaakKMT2026}. The edges represent the complexity when combining properties at the two end points, depicting our novel contributions.}\label{fig:results}
\end{figure}
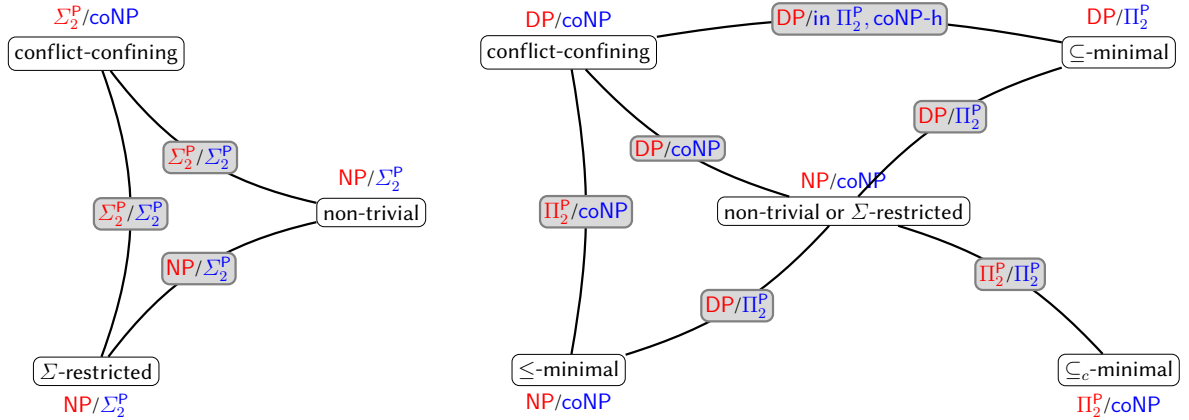

We recently proposed and investigated a range of desirable properties for
abductive hypotheses under repair-based semantics. Some of these have been
considered also under classical semantics, such as \emph{non-triviality} (the
hypothesis should not contain the observation)~\cite{Elsenbroich2006},
\emph{subset} and \emph{cardinality} minimality~\cite{Calvanese2013}, and
restrictions on the concept and role names to be used
(\emph{signature-restricted hypotheses})~\cite{Koopmann21a}. On the other hand,
the typical requirement of consistency of the hypothesis with the knowledge
base, which is also needed to avoid far-fetched explanations, makes less sense
in an inconsistency-tolerant setting, where the knowledge base may already be
inconsistent from the start. Instead, we may require the hypothesis to be
\emph{conflict-confining}~\cite{du2015towards}, which means that
it does not introduce any new conflicts, or at least
\emph{conflict-minimal}, which means it does not introduce more conflicts than
necessary. In a recent conference publication~\cite{HaakKMT2026}, complementing
the research by \citeauthor{du2015towards}~\cite{du2015towards}
that focused on conflict-confining
hypotheses under IAR semantics, we analysed the computational complexity of
abduction for each of these properties, considering brave and AR semantics for
\DLLite and \ELbot. In particular, we looked at the verification problem (is a
given set of assertions a hypothesis satisfying the requirement), and the existence
problem (is there any such hypothesis). An overview of our earlier
results~\cite{HaakKMT2026} is shown in \Cref{fig:results}. This
investigation considered the different properties in isolation, and left open
the question of what happens if one desires more than
one property to hold, for instance, if the hypothesis should be both
non-trivial and conflict-confining, or subset-minimal among all the
signature-restricted hypothesis. These and other combinations are very natural,
which is why we have a closer look at them in this paper. This time, we focus on
\ELbot. Specifically, we again investigate the theoretical complexity
of the verification and the existence problem for hypotheses that satisfy two
of the properties in \Cref{fig:results}, and that combine one property
with one of the minimality requirements.

\section{Preliminaries}

For a general introduction to description logics, we refer the reader to \citeauthor{DBLP:books/daglib/0041477} (\citeyear{DBLP:books/daglib/0041477}).
In this paper, \emph{ABox assertions} refer to \emph{concept assertions} of the form $A(a)$ and \emph{role assertions} of the form $r(a,b)$ for a concept name $A$, role name $r$ and individuals $a,b$.
In other words, we only use concept assertions without complex concepts, which are also known as \emph{flat} or \emph{simple} assertions.
We assume familiarity with computational complexity~\cite{DBLP:books/daglib/0092426}, in particular with the complexity classes
$\NL, \Ptime, \NP, \co\NP$, $\SigmaP$ and~$\PiP$, as well as
$\DP$, the class of decision problems representable as the intersection of a problem in $\NP$ and a problem in $\co\NP$.

\subsection{Repair Semantics}

If a knowledge base is inconsistent, repair semantics can \enquote{restore} consistent versions and admit meaningful reasoning again.
We focus here on ABox repairs and define these as well as two common kinds of repair semantics next.
We will often use the following notions of consistency and inconsistency w.r.t.\ a TBox \calT.
An ABox \calA is \emph{\calT-consistent}, if $\tup{\calT,\calA}\not\models\bot$, and \emph{\calT-inconsistent} otherwise.
Further, we say that \calA is a \emph{\calT-support} of a concept assertion $\alpha$, if $\tup{\calT, \calA} \models \alpha$.

Let $\calK = \tup{\calT, \calA}$ be an inconsistent knowledge base and \query be a Boolean (conjunctive) query.
A \emph{repair} of $\calK$ is a \calT-consistent subset $\calR \subseteq \calA$ and subset-maximal with this property, i.e., there is no \calT-consistent subset $\calR' \subseteq \calA$ that is a strict superset of \calR.
The somewhat dual notion is a \emph{conflict} or \emph{conflict set} \conflict, which is a \calT-inconsistent subset of the ABox and subset-minimal with this property.
We denote by $\Conf(\calK)$ the set of conflicts of \calK.
We recall entailment under \brave~\cite{BienvenuR13} and \ar semantics~\cite{LLRRS-JWS-15}:
\begin{itemize}
	\item $\calK\models_{\brave} \query$ iff there is a repair $\calR$ of $\calK$ s.t.\ $\tup{\calT,\calR} \models \query$.
	\item $\calK\models_{\ar} \query$ iff $\tup{\calT,\calR}\models \query$ for every repair $\calR$ of $\calK$.
\end{itemize}
The complexity of query entailment under repair semantics is well understood~\cite{BienvenuB16}:
For \ELbot, checking entailment of a concept assertion under \brave semantics is \NP-complete in combined complexity, while it is \coNP-complete under \ar semantics.

\subsection{ABox Abduction for Inconsistent KBs}
\label{sec:abd}

We now introduce the necessary notions for ABox abduction under repair semantics,
as introduced in~\cite{HaakKMT2026}.
We begin with the central definitions of abduction problems and hypotheses in this setting.
\begin{definition}\label{def:hypotheses}
  Let $\calK = \tup{\calT, \calA}$ be an inconsistent KB, $\alpha$ a concept assertion
  (called an \emph{observation}) and $\calS \in \{\brave, \ar\}$ such that $\calK \not\models_\calS \alpha$.
  Then, the pair $\tup{\calK, \alpha}$ is called an \emph{$\calS$-abduction problem}.
  A solution for such a problem, called \emph{\calS-hypothesis}, is an ABox \Hyp using only individuals occurring in  \calK and $\alpha$ s.t.\ $\tup{\calT, \calA \cup \Hyp} \models_\calS \alpha$.
  
  Additionally, for a set $\Sigma$ (signature) containing individual, concept and role names, we define the triple $\tup{\calK, \alpha, \Sigma}$ to be a \emph{\SignatureRestricted \calS-abduction problem}.
  A solution for such a problem is an ABox \Hyp using only symbols from $\Sigma$ that is an \calS-hypothesis for $\tup{\calK, \alpha}$.
\end{definition}
Note also that we do not admit fresh individuals beyond the specified signature as
it is done for example in \cite{Koopmann21a}.
Abduction with fresh individuals easily becomes \ExpTime-hard in \ELbot~\cite{Koopmann21a}.
Because there are at most exponentially many repairs, the complexity is in this case dominated by the complexity of (classical) abduction, so such a setting might make
our complexity analysis less insightful.

To obtain hypotheses that are meaningful for explanation purposes, a number of additional properties and minimality criteria for hypotheses that yield \emph{preferred hypotheses} have been introduced.
For two sets $S_1, S_2$ we may abbreviate $\lvert S_1\rvert \leq \lvert S_2\rvert$ by $S_1 \leq S_2$.
We next provide the definition of these properties and notions of minimality.
\begin{definition}\label{def:properties}
  Let ${\calS} \in {\{\brave,\ar\}}$, $\tup{\calK, \alpha}$ an \calS-abduction problem, and 
  ${\preceq} \in {\{\subseteq, \leq\}}$.
  An ABox \Hyp is
  \begin{enumerate}
    \item\label{def:conf-conf} \emph{conflict-confining for $\calK = \tup{\calT, \calA}$}, provided that
    \mbox{$\Conf(\tup{\calT, \calA \cup \Hyp}) = \Conf(\calK)$}.
  \end{enumerate}
  If \Hyp is an \calS-hypothesis for $\tup{\calK, \alpha}$, we call it
  \begin{enumerate}
    \setcounter{enumi}{1}
    \item\label{def:non-triv} \emph{non-trivial}, if $\alpha \not\in \Hyp$,
    \item\label{def:hyp-minimal} \emph{$\preceq$-minimal}, if there is no $\calS$-hypothesis $\Hyp'$ for $\tup{\calK, \alpha}$ s.t.\ $\Hyp' \prec \Hyp$, and
    \item\label{def:conf-minimal}  \emph{$\preceq_c$-minimal}, if there is no \calS-hypothesis $\Hyp'$ for $\tup{\calK, \alpha}$ s.t.\ $\Conf(\tup{\calT, \calA \cup \Hyp'}) \prec \Conf(\tup{\calT, \calA \cup \Hyp})$.
  \end{enumerate}
\end{definition}
Conflict-confinement can be equivalently defined by requiring that
$\tup{\calT, \calR \cup \Hyp} \not\models\bot$ for every repair~\calR of~\calK.
Note that hypotheses introducing new conflicts can be desirable, as erroneous facts might need to implicitly be replaced by new facts, leading to conflicts when considering both together. 

We use the terms \emph{subset-minimal} for $\subseteq$-minimal and \emph{cardinality-minimal} for $\leq$-minimal.
Further, we also consider minimal variants of hypotheses with additional properties:
For example, a $\subseteq$-minimal conflict-confining \ar-hypothesis need only be $\subseteq$-minimal among all conflict-confining \ar-hypotheses.

We investigate the following reasoning problems for a given (\SignatureRestricted) $\Sem$-abduction problem.
\begin{definition}[Reasoning Problems]\label{def:problems}\ 
	Given a (\SignatureRestricted) \calS-abduction problem
	$\mathfrak{A}=\tup{\calK, \alpha}$ \mbox{($\mathfrak{A}=\tup{\calK, \alpha, \Sigma}$)},
	\begin{enumerate}[topsep=-\parskip+\lineskip]
		\item the \emph{existence problem} asks whether $\mathfrak{A}$ has a solution, and
		\item the \emph{verification problem} asks whether a given ABox \calH is a hypothesis for
		$\mathfrak{A}$ (over $\Sigma$).
	\end{enumerate}
\end{definition}

We restate the following two results that were established in~\cite{HaakKMT2026}.

\begin{lemma}\label{lem:cc-compl}
	Given a KB \calK and an ABox \calH, checking whether \calH is conflict-confining for \calK is \coNP-complete.
\end{lemma}

\begin{lemma}\label{lem:ar-triv}
  Let $\tup{\calK, \alpha}$ be an \ar-abduction problem. The following are equivalent:
  \begin{enumerate*}
    \item there is an \ar-hypothesis for $\tup{\calK, \alpha}$,
    \item $\{\alpha\}$ is a conflict-confining \ar-hypothesis for $\tup{\calK, \alpha}$.
  \end{enumerate*} 
\end{lemma}

\section{Complexity Results for Combinations of Multiple Properties}
We will now establish the complexity of existence and verification of hypotheses for combinations of two of the properties of being \SignatureRestricted, non-trivial and conflict-confining, as well as of one of these properties together with one of the minimality criteria.
Note that the combination of conflict-confining with one of the minimalities regarding the set of conflicts does not make sense, as the set of conflicts for a conflict-confining hypothesis is always minimal.

Regarding the combination of \SignatureRestriction with other properties or minimality criteria, we immediately obtain the following general observation:
any hardness result for the individual properties directly carries over to the case were \SignatureRestriction is required additionally.
To see this, observe that one can always set $\Sigma$ to the (full) signature of the given abduction problem, ensuring that $\Sigma$ does not actually restrict the solutions.
\begin{observation}\label{obs:sigma}
	Let $P$ be a property of hypotheses considered in this paper.
	Then, the existence (resp., verification) of $\calS$-hypotheses with property $P$ can be reduced in polynomial time to the existence (resp., verification) of $\Sigma$-restricted $\calS$-hypotheses with property $P$ for $\calS\in\{\brave,\ar\}$.
\end{observation}

\subsection{Existence}
Here, we consider the existence problem for combinations of the three main properties of hypotheses.
Note that minimality criteria are not interesting for existence, as there is a minimal hypothesis with certain properties iff there is any such hypothesis.

First, we briefly motivate the combination of \SignatureRestriction and non-triviality.
One main motivation for using \SignatureRestriction can be to avoid trivial hypotheses.
However, a non-trivial hypothesis might still use all the symbols from the
observation, just in different combinations.
Consequently, we might lose relevant hypotheses if we use this trick.
The following example illustrates such a situation:
\begin{example}\label{ex:nt+sig}
	Let $\calK\dfn \tup{\calT,\calA}$ be a KB with $\calT\dfn \{A\sqcap B \subsum C, D\sqcap \exists r.C\subsum A\}$ and $\calA\dfn \{B(m),r(m,n)\}$.
	Moreover, let $\alpha\dfn C(m)$ be the observation and $\Sigma \coloneqq \{C,D,m,n\}$ be the signature.
	Here, the trivial hypothesis $\{C(m)\}$ is in fact also \SignatureRestricted.
	The hypothesis $\Hyp_1\dfn\{A(m)\}$ is non-trivial but outside our signature.
	Moreover, $\Hyp_2\dfn \{C(n),D(m)\}$	is non-trivial and uses both $C$ and $m$, that is, both the concept name and the individual from the observation.
\end{example}

Turning to the complexity results, the combination of signature restriction and non-triviality does not affect the overall complexity, as the individual cases exhibit the same computational behavior.

\begin{corollary}\label{cor:existence-nt+sig-el}
	The existence problem for non-trivial \SignatureRestricted $\calS$-hypotheses is
	\begin{enumerate*}[label=(\roman*)]
    \item $\NP$-complete, for $\calS=\brave$, and
	  \item $\SigmaP$-complete, for $\calS=\ar$.
  \end{enumerate*}
\end{corollary}
\begin{proof}
	We recall the membership for the non-trivial case (\Cref{fig:results}). 
	For \brave semantics, existence of a non-trivial \brave-hypothesis for an instance $\tup{\tup{\calT, \calA}, \alpha}$
	can be checked by guessing a candidate hypothesis \Hyp over the signature of
	$\tup{\calK, \alpha}$, as well as a candidate repair
	\calR of $\tup{\calT, \calA \cup \Hyp}$ and
	verifying in polynomial time that \Hyp is non-trivial, \calR is a repair,
	and $\tup{\calT, \calR} \models \alpha$.
	For \ar semantics, we guess a hypothesis \Hyp as before, and verify that $\Hyp$ is non-trivial and $\tup{\calT, \calA \cup \Hyp} \models_\ar \alpha$.
	The latter is \ar-entailment, which can be handled by querying an \NP oracle.
	
	For the combination, note that both sketched approaches are guess-and-check algorithms, so we can simply include a check for \SignatureRestriction of \Hyp.
  Hence, we get the same complexity upper bounds as for the non-trivial setting.
	The hardness in both cases follows from the lower bound for existence of non-trivial hypotheses combined with \Cref{obs:sigma}.
\end{proof}

Interestingly, adding signature restriction, or non-triviality,
on top of conflict-confinement, yields two cases in which both semantics incur the same
``higher'' complexity.
\begin{restatable}{theorem}{ExistNTCC}\label{thm:elbot-exist-nt-cc}
	The existence problem for conflict-confining non-trivial $\calS$-hypotheses is $\SigmaP$-complete for $\calS\in\{\brave,\ar\}$.
	This also holds for the case of conflict-confining \SignatureRestricted \calS-hypotheses and the case of conflict-confining non-trivial \SignatureRestricted \calS-hypotheses.
\end{restatable}
\begin{proof}[Proof idea]
	For membership,
	we guess a hypothesis $\Hyp$ over the signature of $\calK$ and verify that \Hyp is non-trivial, $\tup{\calT,\calA\cup\Hyp}\models_\calS \alpha$, and $\Hyp$ is conflict-confining.
  The first check can easily be done in polynomial time, while the last two checks can be performed via $\NP$-oracle calls, yielding \SigmaP-membership.
  In particular, $\calS$-entailment has $\NP$ (resp., \coNP) complexity for $\calS=\brave$ ($\calS=\ar$), while the complexity for checking conflict-confinement was stated in \cref{lem:cc-compl}.
  
	When adding \SignatureRestriction, we guess $\Hyp$ over the given signature $\Sigma$ (instead of the signature of $\tup{\calK, \alpha})$, and if we drop the requirement for non-triviality, the corresponding check is dropped.
	In both cases, we obtain $\SigmaP$-membership.

	For hardness, we first consider the case of brave semantics.
	The existence of conflict-confining \brave-hypotheses is $\SigmaP$-complete (\Cref{fig:results}) and we observe that the reduction easily extends to our present case.
	For non-trivial, the proof follows by observing that the employed reduction constructs a hypothesis that is always non-trivial (see appendix).
	When adding \SignatureRestriction (with or without non-triviality), hardness is retained due to \Cref{obs:sigma}.

	For AR semantics, we observe that the existence problem for non-trivial hypotheses (\Cref{fig:results}) is $\SigmaP$-complete.
	We argue that the reduction employed for hardness in that result extends to the combination of non-triviality with conflict-confinement.
	This requires us to prove that the considered hypothesis in the reduction is indeed conflict-confining.
  We find it convenient to present the reduction here for completeness.
	For hardness, we reduce from the validity problem for $\exists\forall$-QBFs.
	Let $\Psi = \exists Y\forall Z \; \psi$ be an $\exists\forall$-QBF, where $\psi$ is in DNF, and let $X = Y \cup Z$.
	We construct a TBox \calT using concept names $N = \{A_x, A_{\bar x} \mid x \in X\} \cup \{C, A_\psi\}$.
	Intuitively, \calT expresses that for any term $t \in \psi$, the conjunction of concepts corresponding to the literals in $t$ entails the concept $A_\psi$, which represents satisfaction of $\psi$.
	To this end, we define   $\calK\dfn \tup{\calT,\calA}$, where
	\begin{align*}
		\calT &\dfn \{\,A_x \sqcap A_{\bar x} \subsum \bot \mid x\in X\, \} \cup \left\{\,C\sqcap \bigsqcap_{\ell\in t}A_\ell \subsum A_\psi \mid t\in\psi\,\right\}, \\ 
		\calA &\dfn \{\,A_z(m), A_{\bar z}(m)\mid z\in Z\}
	\end{align*}
	for an individual name $m$.
	Now, $\tup{\calK, A_\psi(m)}$ is the desired abduction problem.
	Observe that $\Hyp_{\text{triv}} \dfn \{A_\psi(m)\}$ is the trivial \ar-hypothesis for $\tup{\calK, \alpha}$.
	However, we are considering existence of a non-trivial hypotheses here, i.e., one that does not contain $A_\psi(m)$.
  Intuitively, any non-trivial hypothesis corresponds to an assignment satisfying the formula $\forall Z \, \psi$, containing assertions for all literals over $Y$ satisfied by the assignment.
  It is not hard to see that adding such hypotheses to \calK does not introduce new conflicts, and that there is no need to include additional assertions, allowing us to prove the following claim.
	\begin{restatable}{claim}{ClaimExistNTCC}\label{claim:nt-verification-el}
		$\Psi$ is true iff $\tup{\calK, \alpha}$ admits a non-trivial conflict-confining \ar-hypothesis.
	\end{restatable}

  Hardness when adding \SignatureRestriction to either of the cases again follows from \Cref{obs:sigma}. 
\end{proof}

\subsection{Verification}
We now turn towards verification.
We first argue that the combinations of the main properties (without minimality criteria) can be handled similar as for existence.
Observe that the verification problem for either signature-restricted, or non-trivial $\calS$-hypotheses admits the same complexity as $\calS$-entailment for $\calS\in\{\brave,\ar\}$.
Given an instance $\tup{\calK,\alpha}$ of $\calS$-abduction problem with $\calK=\tup{\calT,\calA}$,
verifying that an ABox $\Hyp$ is a non-trivial $\calS$-hypothesis requires that
\begin{enumerate*}[label=(\arabic*)]
  \item $\tup{\calT, \calA\cup\Hyp}\models_\calS \alpha$, and
  \item \Hyp is non-trivial.
\end{enumerate*}
Here, the complexity is dominated by the first check, as the second step can be performed in polynomial time.
Similarly, for \SignatureRestriction, step (2) above can be replaced by a check that $\Hyp$ is \SignatureRestricted, which again takes polynomial time.
Further, the same is still true when combining non-triviality with \SignatureRestriction.
The lower bound follows from hardness of verification for (general) \calS-hypotheses and non-trivial \calS-hypotheses combined with \cref{obs:sigma}, yielding the following corollary.
\begin{corollary}
	Verification of non-trivial $\calS$-hypotheses and of \SignatureRestricted \calS-hypotheses is
  \begin{enumerate*}[label=(\roman*)]
	  \item $\NP$-complete for $\calS=\brave$, and
	  \item $\co\NP$-complete for $\calS=\ar$.
  \end{enumerate*}
	Moreover, the same is true when combining both properties, i.e., for verification of non-trivial \SignatureRestricted \calS-hypotheses.
\end{corollary}
It is already known that adding conflict-confinement increases the complexity for verification, but only for brave semantics (\Cref{fig:results}). 
We next prove that adding non-triviality or \SignatureRestriction does not increase the complexity beyond that of checking that a given ABox is a conflict-confining hypothesis.

\begin{theorem}\label{thm:cc-nt-verification-el}
	Verification of non-trivial conflict-confining $\calS$-hypotheses is
  \begin{enumerate*}[label=(\roman*)]
	  \item $\DP$-complete for $\calS=\brave$, and
	  \item $\co\NP$-complete for $\calS=\ar$.
  \end{enumerate*}
	The complexity remains the same when non-triviality is replaced by $\Sigma$-restriction, and when requiring all three properties.
\end{theorem}
  \begin{proof}[Proof idea]
	The hardness applies due to the case of verification for conflict-confining hypotheses since the considered hypotheses in each case are indeed non-trivial (see appendix).
  Hardness when adding \SignatureRestriction to either of the cases again follows from \Cref{obs:sigma}.
  We give a bit more detail for the case of conflict-confining non-trivial hypotheses.
	For brevity, we only outline the proof for \brave semantics, which utilises a reduction from the combination of entailment and non-entailment under \brave semantics to verification of non-trivial conflict-confining \brave-hypotheses.
	Given an instance $\tup{\calK, \alpha_1, \alpha_2}$ for some inconsistent KB \calK, the problem asks whether $\calK\models_\brave \alpha_1$ and $\calK\not\models_\brave \alpha_2$.
	This problem is \DP-complete because the first question is \NP-complete and the second question is \coNP-complete under \brave semantics:
  Membership follows immediately, while only using a single KB in the problem is a slight obstacle to showing hardness.
  To address this, reduce from the variant of the problem with two KBs by first ensuring that the first and the second part of the problem use disjoint sets of names, and then defining the new single KB as the disjoint union of both KBs. 
	For the reduction to the abduction problem at hand, assume w.l.o.g.\ that $\alpha_1$ and $\alpha_2$ both use the same individual, but different concept names, and let $\alpha_1 = A(a)$, $\alpha_2= B(a)$, with $\calK = \tup{\calT,\calA}$.
	We consider fresh concepts $C$ and $D$ to construct a KB $\calK' \coloneqq \tup{\calT',\calA}$, for
	\[\calT' \coloneqq \calT \cup \{C \sqcap A\subsum D, C\sqcap B\subsum \bot\},\]
	an observation $\alpha \coloneqq D(a)$, and a hypothesis $\Hyp \coloneqq \{C(a)\}$.
	It can now be shown that \Hyp is a \brave-hypothesis for $\tup{\calK', \alpha}$ iff $\calK \models_\brave A(a)$, and conflict-confining for $\calK'$ iff $\calK \not\models_\brave B(a)$.
  Observing that $\alpha \not\in \Hyp$, we obtain the following claim, proving correctness.
	\begin{restatable}{claim}{ClaimVerNTCC}\label{claim:el-verif-brave-nt-cc}
		\Hyp is a non-trivial conflict-confining \brave-hypothesis for $\alpha$ in $\calK'$ iff $\calK \models_\brave \alpha_1$ and \mbox{$\calK \not\models_\brave \alpha_2$}.
	\end{restatable}

	We next prove membership in each case.
  Let $\tup{\calK, \alpha}$ be an $\calS$-abduction problem with $\calK = \tup{\calT, \calA}$ and $\Hyp$ an ABox.

	We begin with \brave semantics.
	Observe that $\Hyp$ is a non-trivial conflict-confining $\brave$-hypothesis iff
  \begin{enumerate*}[label=(\arabic*)]
    \item \Hyp is non-trivial,
    \item $\tup{\calT, \calA\cup \Hyp} \models_\brave \alpha$, and
    \item $\Conf(\tup{\calT, \calA \cup \Hyp}) = \Conf(\tup{\calT, \calA})$.
  \end{enumerate*}
	Here, (1) can be easily checked in polynomial time.
	Moreover, 
	(2) is instance checking for \ELbot under brave semantics and hence in \NP, while (3) can be checked in \coNP, since the complement can be checked by guessing a conflict~$\calC$ s.t.\ $\calC\in \Conf(\tup{\calT, \calA \cup \Hyp})$ and $\calC\not\in \Conf(\tup{\calT, \calA})$.
	Concretely, this amounts to guessing a $\calC \subseteq \calA\cup\Hyp$ such that:
  \begin{enumerate*}[label=(\alph*)]
		\item $\calC \not\subseteq \calA$, 
	  \item $\tup{\calT, \calC}\models \bot$, and 
		\item $\tup{\calT, \calC\setminus\{\gamma\}}\not\models \bot$ for any $\gamma\in C$.
  \end{enumerate*}
	Checking whether \calC satisfies (a)--(c) can be done in polynomial time.
	Altogether, (a) and (b) together are in \NP, while (c) is in \coNP. 
	Hence, the problem is contained in \DP.

	The membership proof for brave semantics can be adapted to the cases where \SignatureRestriction is added as a requirement, or replaces non-triviality, by noting that checking \SignatureRestriction on top or instead of non-triviality can also be done in polynomial time.

	We now turn to membership for \ar-semantics.
	First, check non-triviality of $\Hyp$ in polynomial time.
	If this is the case, $\Hyp$ is a non-trivial conflict-confining $\ar$-hypothesis iff
  \begin{enumerate*}[label=(\arabic*)]
    \item $\tup{\calT, \calA\cup \Hyp} \models_\ar \alpha$, and
    \item $\Hyp$ is conflict-confining in $\calK$, i.e., $\Conf(\tup{\calT, \calA \cup \Hyp}) = \Conf(\tup{\calT, \calA})$.
  \end{enumerate*}
	As a result, one can guess a candidate repair $\calR$ and a candidate conflict $\calC$ of $\tup{\calT, \calA \cup \Hyp}$ simultaneously, as a counter example to $\Hyp$ being an $\ar$-hypothesis ($\calR$) and conflict-confining ($\calC$).
	More precisely, one needs to check that \calR is a subset-maximal \calT-consistent subset of \calA (hence a repair) and $\tup{\calT, \calR} \not\models \alpha$, as well as that $\calC \not\in \Conf(\tup{\calT,\calA})$, but $\calC\in \Conf(\tup{\calT, \calA\cup\Hyp})$.
  All these checks can be performed in polynomial time.
	The claimed \coNP-membership follows by observing that $\Hyp$ is a conflict-confining hypothesis for $\alpha$ in $\calK$ iff both these checks return false.
  Membership when adding \SignatureRestriction to the conflict-confining and the conflict-confining non-trivial case again follows by noting that \SignatureRestriction can be checked separately in polynomial time.
\end{proof}

\subsubsection*{Combination of Properties with Minimality Criteria}
We next turn towards combinations involving minimality criteria, beginning with $\subseteq$-minimality.
Here, the following general observation regarding our main properties is very helpful.
\begin{observation}\label{obs:downward}
  The properties of being conflict-confining, non-trivial, and \SignatureRestricted are downwards-closed.
  More precisely, let $\tup{\calK, \alpha}$ be an \calS-abduction problem for $\calS \in \{\brave, \ar\}$, \Hyp an ABox and $\Hyp' \subseteq \Hyp$.
  If \Hyp is conflict-confining for \calK, then so is $\Hyp'$.
  Similarly, if $\alpha \not\in \Hyp$, then $\alpha \not\in \Hyp'$, and if \Hyp is an ABox over a signature $\Sigma$, then so is $\Hyp'$.
\end{observation}
This implies that an ABox \Hyp being a $\subseteq$-minimal conflict-confining \calS-hypotheses is equivalent to $\Hyp$ being a $\subseteq$-minimal \calS-hypothesis and being conflict-confining (as a separate requirement). The same observation applies to the other two properties as well as combinations of properties.

Using the established results on verification of $\subseteq$-minimal hypotheses and the above observation, we obtain the following complexities when additionally requiring other properties.
\begin{theorem}\label{cor:verification-nt-min-el}
	Verification of $\subseteq$-minimal non-trivial $\calS$-hypotheses is 
  \begin{enumerate*}[label=(\roman*)]
	  \item $\DP$-complete, for $\calS=\brave$, and
	  \item $\PiP$-complete, for $\calS=\ar$.
  \end{enumerate*}
	The complexity remains same when we replace non-triviality by \SignatureRestriction, and when combining both properties.
\end{theorem}
\begin{proof}[Proof idea]
	The hardness in each case follows from the corresponding result on verification of $\subseteq$-minimal $\calS$-hypotheses (see appendix).
	Observe that the reduction in each case uses a non-trivial hypothesis,
	thus the hardness applies to verification of non-trivial subset-minimal hypotheses.
	Again, when adding \SignatureRestriction (either instead of or in combination with non-triviality), hardness is retained by \Cref{obs:sigma}.
	We outline the reduction for $\DP$-hardness with \brave semantics as it it also used in later results.
	We reduce again from a combination of entailment and non-entailment under \brave semantics to verification of $\subseteq$-minimal \brave-hypotheses, similar to hardness under \brave semantics in the proof of \cref{thm:cc-nt-verification-el}.
	Given an instance $\tup{\calK, \alpha_1, \alpha_2}$ for some inconsistent KB \calK, let $\alpha_1 = A(a)$, $\alpha_2= B(a)$, and $\calK = \tup{\calT,\calA}$.
	We now construct a KB $\calK'$, an observation $\alpha$, and an ABox \Hyp as follows.
	Let $\calK' \coloneqq \tup{\calT',\calA}$ with
	\[\calT' \coloneqq \calT \cup \{B \subsum C, C\sqcap D\sqcap A \subsum Q\},\]
	$\alpha \coloneqq Q(a)$, and $\Hyp \coloneqq \{C(a), D(a)\}$ for fresh concepts $C,D,Q$.
	Intuitively, $\Hyp$ is a Brave-hypothesis for $\tup{\calK', \alpha}$ iff $\calK \models_\brave A(a)$ and $\Hyp$ is $\subseteq$-minimal iff $\calK \not\models_\brave B(a)$.
  To see the latter, note that $\subseteq$-minimality is violated iff $\{D(a)\}$ is also a \brave-hypothesis for $\tup{\calK', \alpha}$, which is the case iff $\calK \models_\brave B(a)$.
  The following claim shows correctness.

	\begin{restatable}{claim}{ClaimVerSubsetNT}\label{claim:nt-subset-verification} 
	\Hyp is a $\subseteq$-minimal non-trivial hypothesis for $\tup{\calK', \alpha}$ iff $\calK \models_\brave \alpha_1$ and $\calK \not\models_\brave \alpha_2$.
	\end{restatable}

  Membership follows from the the upper bound for verification of $\subseteq$-minimal \calS-hypotheses together with \cref{obs:downward}:
  By the observation, it is sufficient to check that \Hyp is a $\subseteq$-minimal \calS-hypotheses, and separately checking that it is non-trivial or \SignatureRestricted.
  Since the latter two properties can be checked in \Ptime, we obtain the desired upper bounds.
  For this, note that both \DP and \PiP are closed under intersection with languages from \Ptime.

\end{proof}

Verification of $\subseteq$-minimal conflict-confining hypotheses enjoys the same complexity for \brave-hypotheses, but the precise complexity of \ar-hypotheses remains open for now.
\begin{theorem}\label{thm:verification-subset-min-el}
	Verification of $\subseteq$-minimal conflict-confining $\calS$-hypotheses is
  \begin{enumerate*}[label=(\roman*)]
	  \item $\DP$-complete for $\calS=\brave$, while it is
    \item in $\PiP$ and \coNP-hard for $\calS=\ar$.
  \end{enumerate*}
  The same holds when additionally requiring non-triviality, \SignatureRestriction, or both.
\end{theorem}
\begin{proof}[Proof idea]
  Membership for both semantics follows from the membership of verification of $\subseteq$-minimal \calS-hypotheses, again using \cref{obs:downward}:
  By the observation, it is sufficient to verify that \Hyp is a $\subseteq$-minimal \calS-hypothesis, and separately checking that it is conflict-confining.
  As the latter is in \coNP and both \DP as well as \PiP are closed under intersection with languages from \coNP, we obtain the desired upper bounds.
  As both non-triviality and \SignatureRestriction can even be checked in \Ptime, the same upper bound still holds when adding one or both of these properties.

  $\DP$-hardness for \brave semantics follows due to the same complexity of verification of $\subseteq$-minimal \brave-hypotheses. We note that the considered hypothesis in the reduction employed for this case in \Cref{cor:verification-nt-min-el} is actually conflict-confining. As a result, \Cref{claim:nt-subset-verification} can be reused after replacing non-triviality by conflict-confinement.
  As argued, the constructed hypothesis is in fact both non-trivial and conflict-confining, thus the hardness also applies when adding non-triviality in the combination.
  Adding \SignatureRestriction can again be handled using \cref{obs:sigma}.

	Finally, \coNP-hardness for \ar semantics follows from the reduction employed in the \coNP-hardness proof for verification of conflict-confining \ar-hypotheses \cite{HaakKMT2026}.
  Precisely, reduce from \ar-entailment as follows.
  Given a KB $\calK = \tup{\calT, \calA}$ and a concept assertion $A(a)$, define the new KB $\calK' = \tup{\calT', \calA}$, where $\calT' \coloneqq \calT \cup \{X \sqcap A \sqsubseteq C\}$, and let $\alpha \coloneqq C(a)$ and $\Hyp \coloneqq \{X(a)\}$.
  We have $\calK \models_\ar A(a)$ iff $\tup{\calT', \calA \cup \Hyp} \models_\ar \alpha$.
  Furthermore, it is easy to see that \Hyp is conflict-confining and minimality follows from $|\Hyp| = 1$.
  Further, \Hyp is non-trivial, while \SignatureRestriction can again be handled using \cref{obs:sigma}, yielding hardness when adding one or both of these properties.
\end{proof}

Cardinality-minimality behaves similar to subset-minimality, if we consider either non-trivial or \SignatureRestricted hypotheses.
It was shown in our earlier work~\cite{HaakKMT2026} that $\leq$-minimal \calS-hypotheses are of size $1$, since there exists some \calS-hypothesis iff the (singleton) trivial hypothesis is one.
This does not apply when considering $\leq$-minimal non-trivial (resp., \SignatureRestricted) hypotheses, as the trivial hypothesis is now disallowed (resp., can be disallowed).
This explains the increased complexity in both cases when compared to the base case (verification of $\leq$-minimal $\calS$-hypotheses).

\begin{theorem}\label{cor:verification-nt-min-card-el}
	Verification of $\leq$-minimal non-trivial $\calS$-hypotheses is 
  \begin{enumerate*}[label=(\roman*)]
	  \item $\DP$-complete for $\calS=\brave$, and
	  \item {$\PiP$}-complete for $\calS=\ar$.
  \end{enumerate*}
	The complexity remains the same when we replace non-triviality by $\Sigma$-restriction, and when combining both properties.
\end{theorem}
\begin{proof}
	Brave semantics:
	For membership, one can guess a set $\Hyp'$ of assertions over the signature of $\calK$ such that $\Hyp'$ is a counter-witness to $\Hyp$ being a $\leq$-minimal $\calS$-hypothesis for $\tup{\calK, \alpha}$.
	Precisely, $\Hyp$ is a $\leq$-minimal non-trivial \brave-hypothesis iff
  \begin{enumerate*}[label=(\arabic*)]
    \item it is a non-trivial \brave-hypothesis, and
    \item there is no $\Hyp'$ over the signature of $\calK$ with $\alpha\not\in \Hyp'$ and $|\Hyp'| < |\Hyp|$ s.t.\ $\tup{\calT,\calA\cup\Hyp'}\models_\brave \alpha$.
  \end{enumerate*}
	Here, (1) can be verified in \NP, while for (2), we guess as a counter-example both a set of assertions $\Hyp'$ of size at most $|\Hyp|-1$ over the signature of \calK as well as a candidate repair $\calR \subseteq \calA \cup \Hyp'$.
  We then verify that $\alpha\not\in\Hyp'$, \calR is a repair, and $\tup{\calT, \calR} \models \alpha$, which can all be checked in polynomial time.
	This shows that (2) can be checked in \coNP, yielding \DP-membership in total.
	The membership also applies when requiring \SignatureRestriction instead of (or in addition to) non-triviality, as verification of (non-trivial) \SignatureRestricted \brave-hypotheses is again in \NP, and instead of (or in addition to) checking non-triviality $\Hyp'$, we can check that $\Hyp'$ only uses names from $\Sigma$ in polynomial time.
	
	For hardness, we reconsider the reduction employed in the case of verification of $\subseteq$-minimal \brave-hypotheses (\Cref{cor:verification-nt-min-el}).
	We observe that the considered hypothesis $\Hyp$ is $\subseteq$-minimal iff $\Hyp$ is $\leq$-minimal. 
	The reduction already applies for the non-trivial case as $\Hyp$ is non-trivial.
  The case of \SignatureRestriction can again be handled using \cref{obs:sigma}.
	
	\ar semantics: The $\PiP$-membership can be shown similarly to \DP-membership for \brave semantics above.
  Precisely, $\Hyp$ is a non-trivial $\leq$-minimal \ar-hypothesis iff
  \begin{enumerate*}[label=(\arabic*)]
    \item it is a non-trivial \ar-hypothesis, and
    \item there is no $\Hyp'$ over the signature of $\calK$ with $\alpha\not\in \Hyp'$ and $|\Hyp'| < |\Hyp|$ s.t.\ $\tup{\calT,\calA\cup\Hyp'}\models_\ar \alpha$.
  \end{enumerate*}
  Here, (1) can be checked in \coNP, while $\Hyp'$ can be guessed as before.
  The main difference is that we do not guess a candidate repair, and instead check $\tup{\calT, \calA \cup \Hyp'} \models_\ar \alpha$ using a \coNP oracle.
  The cases obtained by replacing non-triviality by \SignatureRestriction or combining both can be handled in the same way as above for \brave semantics.

	For \PiP-hardness, we reduce from the complement of the existence problem for non-trivial \ar-hypotheses, which is $\SigmaP$-complete.
	Let $\tup{\calK,\alpha}$ be an \ar-abduction problem over some signature $\Sigma_{\calK}$ and $\alpha\dfn C(m)$.
	We consider $\calA_\Sigma\dfn \{A(a)\mid A,a\in\Sigma_\calK\}\cup \{r(a,b)\mid r,a,b\in\Sigma_\calK\}$ as the set of all the ABox assertions over $\Sigma_\calK$.
	It holds that $\Hyp\subseteq\calA_\Sigma$ for any \ar-hypothesis $\Hyp$ for $\tup{\calK, \alpha}$.
	As a result, any such hypothesis for $\alpha$ must have its size bounded by $N=|\calA_\Sigma|$.
	To complete the reduction, we add fresh concept names $\{X_i \mid i\leq N+1, X_i\not\in\Sigma_\calK\}$, and define the new KB $\calK'\dfn \tup{\calT',\calA}$, where $\calT' \dfn \calT\cup\{ \bigsqcap_{1 \leq i \leq N+1} X_i \subsum C\}$.
%
%
%
	Then, we use the observation $\alpha\dfn C(m)$ as before and define $\Hyp\dfn \{\,X_i(m)\mid i\leq N+1\,\}$. 

	For correctness, we observe that $\Hyp$ is a non-trivial \ar-hypothesis for $\tup{\calK', C(m)}$ of size $N+1$, whereas any non-trivial \ar-hypothesis for $\tup{\calK, \alpha}$ has size $\leq N$ (since $|\calA_\Sigma|=N$).
	This implies the following equivalences:
  $\Hyp$ is not a $\leq$-minimal \ar-hypothesis for $\tup{\calK', C(m)}$ iff there is a non-trivial \ar-hypothesis for $\tup{\calK', \alpha}$ of size $\leq N$ iff there is a non-trivial \ar-hypothesis for $\tup{\calK, \alpha}$.
	This yields the mentioned $\PiP$-hardness for the \ar semantics.
  The case of \SignatureRestriction can again be handled using \cref{obs:sigma}.
\end{proof}
%

Considering conflict-confinement with $\leq$-minimality, the complexity under \ar semantics stays the same as for $\leq$-minimality alone, which follows from \cref{lem:ar-triv}.
In contrast, the complexity again increases for \brave semantics.
Similar to before, this comes from the fact that we do not have the trivial size bound of $1$ in this case, and is closely related to existence of conflict-confining \brave-hypotheses.

\begin{theorem}\label{cor:verification-cc-min-el}
	Verification of $\leq$-minimal conflict-confining $\calS$-hypotheses is 
	\begin{enumerate*}[label=(\roman*)]
    \item \PiP-complete for $\calS=\brave$, and
	  \item $\coNP$-complete for $\calS=\ar$.
  \end{enumerate*}
\end{theorem}
  \begin{proof}[Proof idea]
	For \ar-semantics, \cref{lem:ar-triv} implies that a given \ar-hypothesis is $\leq$-minimal iff it is of size $1$.
  The complexity now readily follows from \coNP-completeness of \ar entailment. 	
	
  We now turn to \brave semantics.	
  For membership, we need to determine whether
  \begin{enumerate*}[label=(\arabic*)]
    \item $\tup{\calT,\calA\cup\Hyp}\models_\brave \alpha$,
    \item $\Hyp$~is conflict-confining, and
    \item there is no $\Hyp'$ over the signature of \calK with $|\Hyp'| < |\Hyp|$ that is conflict-confining for \calK and satisfies $\tup{\calT,\calA\cup\Hyp'} \models_\brave \alpha$.
  \end{enumerate*}
	Here, one can guess the ABox $\Hyp'$ as a counter-witness, and handle all the remaining checks needed in (1)--(3) via oracle calls.
	Consequently, we obtain membership in \PiP.
		
	For hardness, we reduce from the complement of the existence problem for conflict-confining \brave-hypotheses.
	To this aim, we extend and reuse the reduction for this case from our earlier work~\cite{HaakKMT2026}.
	The trick we employ here (details in the appendix) is to use candidate hypothesis $\Hyp$ of very large size following a similar idea as in the proof of \cref{cor:verification-nt-min-card-el} for the case of \ar semantics.
	However, this trick cannot be applied as directly: We know that there might exist a conflict-confining hypothesis even though the observation itself is not conflict-confining.
  In this case, \Hyp would also not be conflict-confining, interfering with the rest of the construction.
	Nevertheless, we invented a way to construct a hypothesis that avoids new conflicts \cite{HaakKMT2026}.
	We apply the same strategy to our considered hypothesis $\Hyp$ here, which requires slightly changing the newly added TBox axiom in the proof of \Cref{cor:verification-nt-min-card-el}.
	This yields the mentioned $\PiP$-hardness.
\end{proof}
%

%

Finally, we consider $\subseteq_c$-minimality. Recall that the combination of $\subseteq_c$-minimality and conflict-confining is not interesting, since conflict-confining hypotheses are always $\subseteq_c$-minimal. 
We consider the remaining cases.
%
Here, we have to determine whether a given non-trivial (or $\Sigma$-restricted) hypothesis is $\subseteq_c$-minimal among all \textit{such} hypotheses. 
Considering \ar semantics, although we know that the observation $\alpha$ has to be a conflict-confining \ar-hypothesis in order for there to exist any \ar-hypothesis by \cref{lem:ar-triv}, we do not know whether and how non-trivial ($\Sigma$-restricted) hypotheses create conflicts. 
Concretely, while it might appear counter-intuitive, \ar-hypotheses are not necessarily conflict-confining.
The following example highlights this observation. 
\begin{example}\label{ex:AR}
	Let $\calK\dfn \tup{\calT,\calA}$ where,
	\begin{align*}
		\calT \dfn \{ & A_1\sqcap B_1 \subsum \bot,\quad A_1\sqcap B_2 \subsum \bot,\quad A_2\sqcap B_1 \subsum \bot,\quad A_2\sqcap B_2\subsum \bot, \\ 
		& A_1\sqcap A_2 \subsum C, \quad B_1\sqcap B_2\subsum C \},\\
		\calA\dfn \{ & A_1(m), B_1(m)\}.
	\end{align*}
	Now, let $\alpha\dfn C(m)$ be our observation.
	Then, the ABox $\Hyp\dfn \{A_2(m), B_2(m)\}$ is not conflict-confining in~$\calK$ since, for example, $\Hyp$ itself is a new conflict in $\calK$.
	Nevertheless, $\tup{\calT,\calA\cup\Hyp} \models_\ar \alpha$, as the only repairs of $\tup{\calT,\calA\cup\Hyp}$ are $\{A_1(m),A_2(m)\}$ and $\{B_1(m),B_2(m)\}$.
  Further, any non-trivial \ar-hypothesis for $\tup{\calK, C(m)}$ is a superset of \Hyp, highlighting that non-trivial \ar-hypotheses might need to introduce new conflicts.
\end{example}

\begin{theorem}\label{thm:verification-nt-min-conf-el}
	Verification of $\subseteq_c$-minimal non-trivial $\calS$-hypotheses is
	$\PiP$-complete for $\calS\in\{\brave,\ar\}$.
	The complexity remains the same when we replace non-triviality by \SignatureRestriction or combine both.
\end{theorem}
\begin{proof}
	We begin with \brave semantics.
	For membership, one can guess as a counter-witness a set of assertions $\Hyp'$ over the signature of $\tup{\calK, \alpha}$ and verify that $\Hyp'$ is a non-trivial \brave-hypothesis for $\alpha$ and $\Hyp'$ admits fewer new conflicts than $\Hyp$ does.
	Precisely, $\Hyp$ is a $\subseteq_c$-minimal non-trivial \brave-hypothesis for $\tup{\calK, \alpha}$ iff 
	\begin{enumerate*}[label=(\arabic*)]
    \item \Hyp is a non-trivial \brave-hypothesis for $\tup{\calK, \alpha}$, and
    \item there is no ABox $\Hyp'$ over the signature of $\tup{\calK, \alpha}$ with $\alpha\not\in\Hyp'$ and $\tup{\calT,\calA\cup\Hyp'}\models_\brave\alpha$ s.t.\ $\Conf(\tup{\calT,\calA\cup\Hyp'}) \subsetneq \Conf(\tup{\calT,\calA\cup\Hyp})$. 
  \end{enumerate*}
  For (2), we can guess a counter-witness $\Hyp'$, while the remaining checks can be handled via oracle calls.
  To compare the sets of conflicts, we can check that every conflict in the first set is contained in the second and that there is a conflict in the second set that does not occur in the first, which can be checked in \coNP and \NP, respectively.
	This results in $\PiP$-membership.
	For \SignatureRestriction, we check whether $\Hyp$ and $\Hyp'$ only use names over $\Sigma$ instead of (or in addition to) checking non-triviality.

	For hardness we reuse the same reduction as for the case of verification of $\subseteq_c$-minimal hypotheses (\Cref{fig:results}).
	Here, we observe that the reduction employs a non-trivial hypothesis (see appendix for details).
  Adding \SignatureRestriction can again be handled using \cref{obs:sigma}.
	
	We now turn to \ar semantics.
	Membership follows by the same argument as for \brave semantics, since the verification of non-trivial \ar-hypotheses and the additional checks can again be handled via oracle calls after guessing a	counter-witness $\Hyp'$.
	For hardness, we incorporate the idea  from \Cref{ex:AR} into	our reduction for the existence of non-trivial conflict-confining \ar-hypotheses from \Cref{thm:elbot-exist-nt-cc}.
	Let $\tup{\Kmc,\alpha}$ be the \ar-abduction problem constructed from an $\exists\forall$-QBF $\Psi$
	as in the proof for \Cref{thm:elbot-exist-nt-cc}.
	We let $A_1,A_2,B_1,B_2$ denote fresh concepts and consider the KB $\calK_e\dfn \tup{\calT_e,\calA_e}$ as constructed in \Cref{ex:AR}.
	To complete the reduction, we set $\calK'\dfn \tup{\calT',\calA'}$ where
	$\calT'\dfn \calT\cup \calT_e$ and $\calA'\dfn \calA\cup \calA_e$ and use the same observation $\alpha$.
	Finally, we take the hypothesis $\Hyp_e\dfn \{A_2(m),B_2(m)\}$ as in the example, which creates precisely three new conflicts $\{\{A_1(m),B_2(m)\}$, $\{A_2(m),B_1(m)\}$, and $\Hyp_e$.
	Then, the correctness follows from the following claim.
	\begin{restatable}{claim}{ClaimVerCminSubsetNT}
		$\Psi$ is true iff $\Hyp_e$ is not a non-trivial $\subseteq_c$-minimal \ar-hypothesis for $\tup{\calK',\alpha}$.
	\end{restatable}
\begin{claimproof}
	Recall that $\Psi$ is true iff $\tup{\calK,\alpha}$ admits a non-trivial conflict-confining \ar-hypothesis (\Cref{claim:nt-verification-el}).
	We observe that $\tup{\calK,\alpha}$ admits a non-trivial conflict-confining \ar-hypothesis iff $\tup{\calK',\alpha}$ admits a non-trivial conflict-confining \ar-hypothesis (the same one) iff $\Hyp_e$ is not a $\subseteq_c$-minimal non-trivial \ar-hypothesis for $\tup{\calK',\alpha}$.

	Indeed, if $\Psi$ is false, there is no \ar-hypothesis over the signature of \calK that entails $\alpha$.
	Hence, the only way to obtain the \ar-entailment of $\alpha$ is via using new axioms in $\calK'$. 
	This renders $\Hyp_e$ to be the only \ar-hypothesis for $\alpha$ in $\calK'$, and hence makes it trivially $\subseteq_c$-minimal.
\end{claimproof}%
\end{proof}

The membership in \Cref{thm:verification-nt-min-conf-el} only applies to $\subseteq_c$-minimality and does not seem to extend to $\leq_c$-minimality.
While in the case of the former, the necessary comparison of the sets of conflicts introduced by \Hyp and the guessed ABox $\Hyp'$ can be handled using an \NP oracle,
this is not the case when we need to compare the cardinalities of these sets.
Indeed, the following example illustrates that a relatively small hypothesis can trigger exponentially many new conflicts, hinting at potential $\#\Ptime$-hardness of counting the number of conflicts.
How this affects the complexity for $\leq_c$-minimality is left for future work.
\begin{example}
	Let $\calK\dfn \tup{\calT,\calA}$ where, 
	\begin{align*}
	  \calT &\dfn  \{ C \subsum D \} \cup \{\, A_i \subsum X_i, B_i \subsum X_i \mid i \leq n \,\} \cup \{ D \sqcap X_1 \sqcap \dots \sqcap X_n \subsum \bot\}, \\
	  \calA &\dfn \{\, A_i(m), B_i(m) \mid i\leq n \,\}.
	\end{align*}
	Now, let  $\alpha \dfn D(m)$ be our observation.
	Then, the ABox $\Hyp\dfn \{C(m)\}$ is a \brave hypothesis for $\tup{\calK, \alpha}$ of size one.
	Nevertheless, $\Hyp$ triggers exponentially many new conflicts in $\calK$, each of the form $\{C(m), Y_1(m), Y_2(m), \dots, Y_n(m)\}$ where $Y_i \in \{A_i,B_i\}$ for all $i$.
	Note that $\calK$ is a consistent KB to begin with.
\end{example}

\section{Conclusion and Outlook}
We observe that combining basic properties (excluding minimality criteria) and potentially subset-minimality does not increase the complexity for either existence or verification of \calS-hypotheses in \ELbot:
The complexity for all considered cases is governed by the highest complexity among the individual properties involved.
In contrast, we often see an increase in the complexity, when combining cardinality-minimality or conflict-minimality with other properties.

Observe that the hardness proofs for some of our results (e.g., Theorems~\ref{thm:elbot-exist-nt-cc}, \ref{cor:verification-nt-min-el} and~\ref{thm:verification-subset-min-el}) do not use role assertions.
Consequently, several of our results already hold for propositional Horn logic.
Nevertheless, we leave a detailed complexity characterization for this fragment as future work.

It will be interesting to complete the picture by also extending the study to conflict-minimality when considering the number of conflicts (instead of subset-minimality), and closing the gap for the case of $\subseteq$-minimal conflict-confining \ar-hypotheses, which both remain open for now.

Further, we previously
also obtained complexity results for ABox abduction in the description logic
\DLLite~\cite{HaakKMT2026} .
Preliminary work for combinations of properties in this setting shows that in some cases the complexity is again governed by the highest complexity among the individual properties, but there is a notable exception.
It turns out that in case of \DLLite, conflict-confining \ar-hypotheses have a much simpler structure:
Every conflict-confining \ar-hypothesis contains a \emph{singleton} \ar-hypothesis, which is of course again conflict-confining.
This often makes abduction problems \emph{easier} when adding conflict-confinement in addition to some other property (or minimality condition) compared to the same problem without requiring conflict-confinement.
It will be interesting to fully understand the complexity landscape for combinations of properties in this setting.

It should be noted that in case of expressive DLs such as \ALC and its extensions, the complexity of the considered abduction problems seems significantly less interesting.
The reason is that for these DLs, entailment is often already \ExpTime-complete.
Consequently, the complexity of abduction tends to be dominated by the complexity of entailment.

Beside considering different DLs, another important direction for the future is to consider other repair semantics.
Most prominently, IAR semantics are currently missing from our study.
While there has been previous work on abduction under IAR semantics, it was more practical in nature, obtaining an upper bound for one specific case in \DLLite.
The complexity landscape for different (combinations of) properties remains widely unexplored.

\bibliography{main}


\section{Technical Appendix}
For constructions referenced in the main paper, we provide full proofs of the corresponding results here.

The following example from our previous paper~\cite{HaakKMT2026}
will be useful in some of our reductions
(e.g., when translating \SignatureRestriction to conflict-confinement).

\begin{example}\label{ex:brave-cc}
	Let $\calK = \tup{\calT, \calA}$, where
	\begin{align*}
		\calT &= \{\, A \sqcap B \sqsubseteq \bot, B \sqcap C \sqsubseteq \bot, C \sqcap D \sqsubseteq A \,\}, \text{ and} \\
		\calA &= \{\, B(a), C(a) \,\},
	\end{align*}
	and let $\alpha = A(a)$. 
	It is easy to see that $\{\alpha\}$ is a \brave-hypothesis for $\tup{\calK, \alpha}$, but results in a new conflict $\{A(a),B(a)\}$ of $\tup{\calT,\calA\cup\{\alpha\}}$.
	Hence, $\{\alpha\}$ is not conflict-confining in $\calK$.
	However, $\calH\dfn \{D(a)\}$ entails $A(a)$ in the repair $\{C(a), D(a)\}$ and is consistent with all repairs of \calK. 
	Intuitively, the only repair of $\calA\cup\calH$ where $\alpha$ is entailed, is the one that already got rid of the conflict with $\alpha$.
\end{example}

\subsection{Constructions used for \Cref{thm:elbot-exist-nt-cc}}\label{sec:cc}

\begin{theorem}\label{thm:elbot-exist-brave-cc}
  The existence problem for conflict-confining \brave-hypotheses is \SigmaP-hard.
\end{theorem}

\begin{proof}
We reduce from non-validity of $\forall\exists$-QBFs to existence of conflict-confining \brave-hypotheses.
Let $\Phi\dfn \forall Y\exists Z.\varphi$, where $\varphi$ is a CNF represented as a set of clauses, where clauses are sets of literals.
Let $X \coloneqq Y \cup Z$.
We construct the KB $\calK = \tup{\calT, \calA}$, using concept names $N = \{A_x, A_{\bar x} \mid x \in X\} \cup \{V_y \mid y \in Y\} \cup \{A_c \mid c \in \varphi\} \cup \{A_{\varphi}, A_{\bar\varphi}, C\}$. 
Intuitively, our encoding is based on a translation for the setting of \SignatureRestriction, but we use conflict-confinement to encode this restriction instead by using the idea from \cref{ex:brave-cc}. 
Precisely, we ensure that any hypothesis that uses symbols outside the intended signature $\Sigma = \{\, A_y, A_{\bar y} \mid y \in Y \,\} \cup \{m\}$ introduces new conflicts, and is thus not conflict-confining.
For this, we use additional concept names $C_d$ and $B_d$.

Formally, we define
\begin{align*}
	\calT \coloneqq {} & \{\, C_d\sqcap \bigsqcap_{y\in Y} V_y \sqcap A_{\bar \varphi} \subsum C\, \} \cup {} \\
	%
	& \{\,C_d\sqcap A_y \subsum V_y, C_d\sqcap A_{\bar y} \subsum V_y \mid y\in Y\,\} \cup {} \\
	& \{\,A_x \sqcap A_{\bar x} \subsum \bot \mid x\in X\, \}  \cup {} \\
	%
	& \{\, C_d\sqcap A_\ell \subsum A_c \mid \ell \in c, c\in\varphi\, \}\cup {} \\ 
	& \{\,C_d\sqcap \bigsqcap_{y\in Y} V_y\sqcap \bigsqcap_{c\in \varphi}A_c\subsum A_\varphi\, \} \cup {} \\ 
	& \{\, A_\varphi \sqcap A_{\bar\varphi}\subsum \bot\,\} \cup \{\, A_\varphi\sqcap B_d \subsum \bot\, \}  \cup {} \\ 
	%
	& \{\, C_d\sqcap B_d \subsum \bot\, \} \cup {} 
	\{\, C\sqcap B_d \subsum \bot\, \} \cup {}\\ 
	& \{\, A_c\sqcap B_d \subsum \bot \mid c\in\varphi\, \} \cup {} 
	\{\, V_y\sqcap B_d \subsum \bot\, \} \text{ and} \\
	\calA \coloneqq & \{\, A_z(m), A_{\bar z}(m) \mid z\in Z \,\} \cup \{A_{\bar\varphi}(m), B_d(m), C_d(m)\},
\end{align*}
where $m$ is an individual.
Finally, let $\calK \coloneqq \tup{\calT, \calA}$ and $\alpha \coloneqq C(m)$.

Axioms in each line of $\calT$ encode the following intuition:
(i) and (ii) enforce a complete assignment over $Y$ as a counter-example to the satisfaction of $\varphi$,
(iii) ensures a valid assignment over $X$,
(iv) satisfaction of clauses via their literals,
(v) satisfaction of $\varphi$, but not solely by a partial assignment over $Z$ variables, and 
(vi) satisfaction of $\varphi$ causes a new conflict.
The remaining axioms use the idea from \cref{ex:brave-cc} to implicitly enforce a \SignatureRestriction, disallowing certain assertions in any conflict-confining hypothesis.

Now $\tup{\calK, \alpha}$ is our desired abduction problem.
It is a \brave-abduction problem:
The KB \calK is inconsistent and as it contains no assertions of the form $A_\ell(m)$ for any $\ell\in\{y,\bar y\}$, we have $\calK \not\models C(m)$.

\begin{claim}\label{claim:el-cc-brave}
	$\Phi$ is false iff $\tup{\calK, \alpha}$ admits a non-trivial conflict-confining \brave-hypothesis.
\end{claim}
\begin{claimproof}
	We first observe that the trivial hypothesis $\calH_t\dfn \{C(m)\}$ is not conflict-confining in $\calK$.
	This follows easily by considering the new conflict $\{B_d(m), C(m)\}\in\Conf(\tup{\calT,\calA\cup\calH_t})$.
	Therefore, we only prove the equivalence between satisfying assignments of $\Phi$ and conflict-confining \brave-hypothesis for $\alpha$ in $\calK$.

	Suppose $\Phi$ is false and let $\theta_Y$ (seen as a set of literals over $Y$) be an assignment over $Y$ s.t.\ $\forall Z \varphi[\theta_Y]$ is false, where $\varphi[\theta_Y]$ denotes the formula obtained from $\varphi$ by applying the partial assignment $\theta_Y$.
	Define $\calH \coloneqq \{\, A_\ell(m) \mid \ell \in \theta_Y \,\}$. 
	We now show that \Hyp is a conflict-confining \brave-hypothesis for $\tup{\calK, \alpha}$ by showing (i) $\tup{\calT, \calA\cup\calH} \models_\brave \alpha$ and (ii) $\calH$ is conflict-confining.
	
	For (i), let $\calB \coloneqq \Hyp \cup \{C_d(m), A_{\bar\varphi}(m)\}$.
	It is easy to see that \calB is \calT-consistent, so there is a repair $\calR \supseteq \calB$ of $\tup{\calT, \calA \cup \Hyp}$.
	Further, we have $\tup{\calT, \calB} \models C(m)$ and therefore $\tup{\calT, \calR} \models C(m)$.
	
	To see (ii), we first observe that $\calH$ does not trigger any inconsistency due to an axiom of the form $A_y\sqcap A_{\bar y}\subsum \bot$, as $\theta_Y$ is an assignment.
	Hence, \Hyp is \calT-consistent.
	Next, we show that for any repair $\calR'$ of \calK, the set $\calR' \cup \Hyp$ is \calT-consistent.
	Due to the axioms of the form $A_x \sqcap A_{\bar x} \sqsubseteq \bot$, the set $\calR' \cap \{\, A_z(m), A_{\bar z}(m) \mid z \in Z \,\}$ corresponds to a (potentially partial) assignment over $Z$.
	Since $\varphi[\theta_Y]$ is false for all assignments over $Z$ by assumption, we have $\tup{\calT, \calR' \cup \Hyp} \not\models A_\varphi$ by construction of \calT.
	This means that $A_\varphi \sqcap A_{\bar \varphi} \sqsubseteq \bot$ does not trigger.
	The remaining disjointness axioms use the idea from \cref{ex:brave-cc} to avoid conflicts in $\tup{\calT, \calR' \cup \Hyp}$:
	Due to the axiom $C_d \sqcap B_d \sqsubseteq \bot$, we have $C_d(m) \not\in \calR'$ or $B_d(m) \not\in \calR'$.
	In both cases, none of the remaining disjointness axioms can trigger, as either the assertion $B_d(m)$ is missing (and not entailed) or no assertion of the form $C(m)$, $A_\varphi(m)$, $A_c(m)$, or $V_y(m)$ is entailed from $\tup{\calT, \calR' \cup \Hyp}$.
	
	Conversely, suppose that there is a conflict-confining \brave-hypothesis \Hyp for $\tup{\calK, \alpha}$.
	We first argue that, w.l.o.g., \Hyp only contains assertions of the form $A_y(m)$ and $A_{\bar y}(m)$.
	To this end, note that \Hyp may only contain assertions over the individual $m$, since it is the only individual in $\tup{\calK, \alpha}$.
	Further, we can assume w.l.o.g.\ that \Hyp only contains assertions using concept or role names occurring in \calK, as fresh concept and role names cannot help entailment of $\alpha$ in \calT.
	Additionally we can assume w.l.o.g.\ that it does not contain any assertions already present in \calA.
	Finally, we observe that all other assertions over the signature of $\tup{\calK, \alpha}$ would introduce new conflicts, contradicting the assumption that \Hyp is conflict-confining:
	The assertion $C(m)$ would introduce the new conflict $\{C(m), B_d(m)\}$, while the assertion $A_\varphi(m)$ would introduce the new conflict $\{A_\varphi(m), B_d(m)\}$.
	Any assertion $A_c(m)$ for $c \in \varphi$ would introduce the new conflict $\{A_c(m), B_d(m)\}$, and any assertion $V_y(m)$ for $y \in Y$ would introduce the new conflict $\{V_y(m), B_d(m)\}$.
	
	Next, it is easy to see that \Hyp contains exactly one of the assertions $A_y(m)$ and $A_{\bar y}(m)$ for each $y \in Y$:
	It contains at least one, since $C(m)$ is entailed in some repair, and this is only possible when $\bigsqcap_{y \in Y} V_y$ is entailed in that repair.
	It contains at most one, since otherwise the axiom $A_y \sqcap A_{\bar y} \sqsubseteq \bot$ would lead to a new conflict.
	Consequently, \Hyp corresponds to an assignment $\theta_Y$ over $Y$.
	
	Now consider an assignment $\theta_Z$ over $Z$, represented as a set of literals.
	We now show that $\varphi[\theta_Y \cup \theta_Z]$ is false.
	As this applies to all assignments over $Z$, it implies that $\forall Y \exists Z \varphi(Y,Z)$ is false.
	Let
	\[\calA_Z \coloneqq \{\, A_\ell(m) \mid \ell \in \theta_Z \,\},\]
	i.e., the subset of \calA corresponding to assignment $\theta_Z$.
	The set $\calA_Z$ is \calT-consistent, since $\theta_Z$ is an assignment.
	As for all $y \in Y$ we have $\tup{\calT, \calA_Z} \not\models V_y(m)$, $\calA_Z \cup \{A_{\bar\varphi}(m)\}$ is \calT-consistent.
	Due to \Hyp being conflict-confining, this implies that $\calA_Z \cup \{A_{\bar \varphi}(m)\} \cup \Hyp$ is also \calT-consistent.
	Consequently, we have $\tup{\calT, \calA_Z \cup \Hyp} \not\models A_\varphi(m)$, implying that $\varphi[\theta_Y \cup \theta_Z]$ is false by construction of \calT.
\end{claimproof}
\end{proof}

\begin{theorem}\label{thm:elbot-exist-ar-nt}
  The existence problem for non-trivial \ar-hypotheses is \SigmaP-hard.
\end{theorem}

We reduce from validity problem for $\exists\forall$-QBFs.
Let $\Psi = \exists Y\forall Z \; \psi$ be a $\exists\forall$-QBF, where $\psi$ is in DNF, and let $X = Y \cup Z$.
We construct a TBox \calT using concept names $N = \{A_x, A_{\bar x} \mid x \in X\} \cup \{C, A_\psi\}$.
Intuitively, \calT expresses that for any term $t \in \psi$, the conjunction of concepts corresponding to the literals in $t$ entails the concept $A_\psi$, which represents satisfaction of $\psi$.
To this end, we define   $\calK\dfn \tup{\calT,\calA}$, where 
\begin{align*}
	\calT \dfn {} &\{\,A_x \sqcap A_{\bar x} \subsum \bot \mid x\in X\, \} \cup {} \\ 
	& \left\{\,C\sqcap \bigsqcap_{\ell\in t}A_\ell \subsum A_\psi \mid t\in\psi\,\right\}, \\ 
	\calA \coloneqq {} & \{\,A_z(m), A_{\bar z}(m)\mid z\in Z\}, 
\end{align*}
for an individual name $m$.
Here, the first set of axioms ensures that repairs encode (potentially partial) assignments over $X$ and the second encodes that $\psi$ is satisfied iff at least one its terms is satisfied.
Now, $\tup{\calK, A_\psi(m)}$ is the desired abduction problem.
It is a valid abduction instance, as \calK is obviously inconsistent, and $\calK \not\models_{\ar} A_\psi(m)$, as \calK does not contain the assertion $C(m)$.
Observe that $\calH_{\text{triv}} \dfn \{A_\psi(m)\}$ is the trivial \ar-hypothesis for $\tup{\calK, \alpha}$.
However, we are considering existence of non-trivial hypotheses here, i.e., one that does not contain $A_\psi(m)$.

%
%
%
\ClaimExistNTCC*
%
%
	\begin{claimproof}
		($\Rightarrow$) Let $\theta_Y$ be an assignment over $Y$, represented as a set of literals, witnessing that $\Psi$ is true, i.e., f.a.\ assignments $\theta_Z$ over $Z$, $\theta_Y \cup \theta_Z \models t$ for some $t \in \psi$.
		Define $\Hyp \dfn \{C(m)\} \cup \{\, A_\ell(m) \mid \ell \in \theta_Y \,\}$.
		$\calH$ is indeed non-trivial.
		We next show that $\calH$ is in fact a conflict-confining \ar-hypothesis for $\alpha$ in $\calK$.
%
%
		Consider any repair \calR of $\calK' \coloneqq \tup{\calT, \calA \cup \Hyp}$.
		Since $\theta_Y$ is an assignment, \Hyp is \calT-consistent.
		Even more, no subset of \Hyp is contained in any conflict of $\calK'$ and hence $\calH$ is conflict-confining in $\calK$.
		Moreover, we have $\Hyp \subseteq \calR$.
		Now let $\theta_Z \dfn \{\, \ell \in Z \cup \bar Z \mid A_\ell(m) \in \calR \,\}$.
		We argue that $\theta_Z$ is a (full) assignment over $Z$.
		As \calR is \calT-consistent, at most one of the assertions $A_x(m)$ and $A_{\bar x}(m)$ is in \calR for each $x \in X$ due to the first form of axioms in \calT.
		On the other hand, since \calR is subset-maximal among the \calT-conistent subsets of $\calA \cup \Hyp$, it contains at least one of those to assertions for each $x \in X$.    Hence, there is some term $t \in \psi$ s.t.\ $t \subseteq \theta_Y \cup \theta_Z$ by our assumption on $\theta_Y$.
		Since $\Hyp \subseteq \calR$ and by construction of \calT, this implies that $\tup{\calT, \calR} \models A_\psi(m)$.
		This proves that $\calK' \models_\ar A_\psi(m)$, since the argument applies to all repairs \calR of $\calK'$.
		
		($\Leftarrow$) Let \Hyp be a non-trivial conflict-confining \ar-hypothesis for $\tup{\calK, \alpha}$.
		We first argue that w.l.o.g., \Hyp only contains assertions of the form $A_\ell(m)$ and $C(m)$:
		It can only contain assertions over individual $m$, as it may not use fresh individuals.
		W.l.o.g.\ it does not contain assertions using concept and role names not occuring in $\tup{\calK, \alpha}$ or already present in \calA, as including these assertions in \Hyp does not contribute to the entailment of $A_\psi(m)$.
		Finally, $A_\psi(m) \not\in \Hyp$, as \Hyp is non-trivial.
		
		It is now easy to see that $C(m) \in \Hyp$, as otherwise $A_\psi(m)$ could not be entailed.
		Further, \Hyp contains at most one of the assertions $A_y(m)$ and $A_{\bar y}(m)$ for each $y \in Y$ as $\calH$ is conflict-confining (containing both assertions for some $y\in Y$ would create a new conflict due to $A_y\sqcap A_{\bar y}\subsum \bot$).
%
		Now, define the (potentially partial) assignment $\theta_Y \dfn \{\, \ell \mid A_\ell(m) \in \calH \,\}$ over $Y$.
		To finish the proof, we show that for each assignment $\theta_Z$ over $Z$, we have $\theta_Y \cup \theta_Z \models \psi$.
		Let $\calR_Z \dfn \{\, A_\ell(m) \mid \ell \in \theta_Y \cup \theta_Z \,\} \cup \{C(m)\}$.
		As $\theta_Y$ and $\theta_Z$ are assignments, it is easy to see that $\calR_Z$ is a repair of $\calK'$.
		Hence, we have $\tup{\calT, \calR_Z}$ by assumption.
		But this implies that $\tup{\calT,\calR_Z} \models A_\psi(m)$ by assumption, and hence there is some term $t \in \psi$, such that $t \subseteq \calR_Z$ by construction of \calT.
		Consequently, $\theta_Y \cup \theta_Z \models \psi$.
		Since $\theta_Z$ is an arbitrary assignment over $Z$, we conclude that $\Psi$ is true.
	\end{claimproof}

\subsection{Constructions used for \Cref{thm:cc-nt-verification-el}}

\begin{theorem}\label{thm:elbot-verif-brave-cc}
  Verification of conflict-confining \brave-hypotheses is \DP-hard.
\end{theorem}

\begin{proof}
We reduce from a combination of entailment and non-entailment under \brave semantics to verification of conflict-confining \brave-hypotheses.
Given an instance $\tup{\calK, \alpha_1, \alpha_2}$ for some inconsistent KB \calK, the problem asks whether $\calK\models_\brave \alpha_1$ and $\calK\not\models_\brave \alpha_2$.
This problem is \DP-complete because the first question is \NP-complete and the second question is \coNP-complete under \brave semantics.
For the reduction, assume w.l.o.g.\ that $\alpha_1$ and $\alpha_2$ are BIQs over the same individual, and let $\alpha_1 = A(a)$, $\alpha_2= B(a)$, and $\calK = \tup{\calT,\calA}$.
We construct a KB $\calK'$, an observation $\alpha$, and a hypothesis \Hyp next.
Let $\calK' \coloneqq \tup{\calT',\calA}$ with
\[\calT' \coloneqq \calT \cup \{C \sqcap A\subsum D, C\sqcap B\subsum \bot\},\]
$\alpha \coloneqq D(a)$, and $\Hyp \coloneqq \{C(a)\}$ for fresh concepts $C$ and $D$.
The instance is a valid abduction problem, since $\calK'$ is inconsistent and $\calK' \not\models_\brave \alpha$.
Intuitively, \Hyp is a Brave-hypothesis for $\tup{\calK', \alpha}$ iff $\calK \models_\brave A(a)$ and \Hyp is conflict-confining for $\calK'$ iff $\calK \not\models_\brave B(a)$.
It remains to show correctness, which is stated in the following claim.

\ClaimVerNTCC*
\begin{claimproof}
	Observe that $\calH\dfn \{C(a)\}$ is non-trivial indeed.
	Hence, we only prove the correctness of the claim regarding conflict-confinement.

	($\Rightarrow$)
	Suppose \Hyp is a conflict-confining \brave-hypothesis for $\tup{\calK', D(a)}$.
	Notice that the only way to obtain the entailment $\calK' \models_\brave D(a)$ is via the TBox axiom $C\sqcap A\subsum D$, since no axiom in \calK contains $D$.
	Therefore, we must have $\calK \models_\brave A(a)$.
	Moreover, we have $\calK \not\models_\brave B(a)$:
	Suppose to the contrary that $\calK \models_\brave B(a)$ and let $\calR$ be a witnessing repair such that $\tup{\calT, \calR} \models B(a)$.
	Then, we have that $\tup{\calT', \calR \cup \calH} \models \bot$, in particular due to the axiom $C \sqcap B \subsum \bot$ and the assertion $C(a)$.
	Since \calR is a repair, and hence \calT-consistent, there must be a conflict of $\tup{\calT', \calR \cup \Hyp}$ that is not a conflict of $\tup{\calT',\calA}$.
	But this leads to a contradiction to our assumption that \Hyp is conflict-confining for $\calK'$.
	As a result, $\calK'\not\models_\brave B(a)$ must be true.
	
	($\Leftarrow$) Suppose $\calK \models_\brave A(a)$ and $\calK \not\models_\brave B(a)$.
	Then, $\calK' \models_\brave A(a)$ and hence
	\[\tup{\calT', \calA\cup\Hyp} \models_\brave D(a).\]
	Therefore \Hyp is indeed a \brave-hypothesis for $\tup{\calK', \alpha}$.
	To show that \Hyp is conflict-confining for $\calK'$, suppose to the contrary that there is a conflict $\calC\in \Conf(\tup{\calT', \calA\cup\Hyp})$ such that $\calC\not\in \Conf(\tup{\calT', \calA})$.
	In particular, this implies that $C(a)\in\calC$ since $\calH=\{C(a)\}$.
	As a result, we have $\calC\setminus\{C(a)\}\subseteq \calA$ and $\calC$ is $\calT'$-consistent, and hence also \calT-consistent.
	But this implies that $\tup{\calT', \calC} \models B(a)$, since the only conflict involving $C(a)$ is via the axiom $C\sqcap B\subsum \bot$.
	Further, this means that $\tup{\calT, \calC} \models B(a)$, as the new axioms in $\calT'$ do not help entailment of $B(a)$.
	As \calC is \calT-consistent, there exists a repair $\calR \supseteq \calC$ of \calK, and we have $\tup{\calT, \calR} \models B(a)$.
	Consequently, $\calK \models_\brave B(a)$.
	But this is a contradiction to our assumption, so \Hyp must be conflict-confining.
\end{claimproof}
\end{proof}

\begin{theorem}\label{thm:elbot-verif-ar-cc}
  Verification of conflict-confining \ar-hypotheses is \coNP-hard.
\end{theorem}

\begin{proof} 
	We reduce from \ar-entailment.
	To achieve this, let $\calK=\tup{\calT,\calA}$ and $A(a)$ be an instance of \ar-entailment.
	Define $\calK' \coloneqq \tup{\calT',\calA}$, where
  \[\calT' \coloneqq \calT \cup \{X \sqcap A \subsum C\}.\]
	Finally, let $\alpha \dfn C(a)$ and $\Hyp \dfn \{X(a)\}$.
	We observe that $\calK \models_\ar A(a)$ iff $\tup{\calT', \calA \cup \Hyp} \models_\ar \alpha$.
	Note that $\Hyp$ is trivially conflict-confining as it uses a fresh concept name $X$ that cannot participate in any conflict.
	
	We now turn to \brave semantics.
	For membership, observe that $\Hyp$ is a conflict-confining $\brave$-hypothesis iff (1) $\tup{\calT, \calA\cup \Hyp} \models_\brave \alpha$, and (2)  $\Conf(\tup{\calT, \calA \cup \Hyp}) = \Conf(\tup{\calT, \calA})$.
	
	In the case of $\brave$-hypotheses, (1) is instance checking for \ELbot and hence in \NP, while (2) can be checked in \coNP by universally guessing a conflict $\calC$ such that $\calC\in \Conf(\tup{\calT, \calA \cup \Hyp})$ and $\calC\not\in \Conf(\tup{\calT, \calA})$.
	The last check can be performed  in polynomial time.
	Hence, the problem is contained in \DP.
	
	For hardness, we reduce from a combination of entailment and non-entailment under \brave semantics to verification of conflict-confining \brave-hypotheses.
	Given an instance $\tup{\calK, \alpha_1, \alpha_2}$ for some inconsistent KB \calK, the problem asks whether $\calK\models_\brave \alpha_1$ and $\calK\not\models_\brave \alpha_2$.
	This problem is \DP-complete because the first question is \NP-complete and the second question is \coNP-complete under \brave semantics.
	For the reduction, assume w.l.o.g.\ that $\alpha_1$ and $\alpha_2$ are BIQs over the same individual, and let $\alpha_1 = A(a)$, $\alpha_2= B(a)$, and $\calK = \tup{\calT,\calA}$.
	We construct a KB $\calK'$, an observation $\alpha$, and a hypothesis \Hyp next.
	Let $\calK' \coloneqq \tup{\calT',\calA}$ with
  \[\calT' \coloneqq \calT \cup \{C \sqcap A\subsum D, C\sqcap B\subsum \bot\},\]
  $\alpha \coloneqq D(a)$, and $\Hyp \coloneqq \{C(a)\}$ for fresh concepts $C$ and $D$.
	The instance is a valid abduction problem, since $\calK'$ is inconsistent and $\calK' \not\models_\brave \alpha$.
	Intuitively, \Hyp is a Brave-hypothesis for $\tup{\calK', \alpha}$ iff $\calK \models_\brave A(a)$ and \Hyp is conflict-confining for $\calK'$ iff $\calK \not\models_\brave B(a)$.
	It remains to show correctness, which is stated in the following claim.
	\begin{claim}\label{claim:el-verif-brave-cc}
		\Hyp is a conflict-confining \brave-hypothesis for $\alpha$ in $\calK'$ iff $\calK \models_\brave \alpha_1$ and $\calK \not\models_\brave \alpha_2$.
	\end{claim}
	\begin{claimproof}
	($\Rightarrow$)
	Suppose \Hyp is a conflict-confining \brave-hypothesis for $\tup{\calK', D(a)}$.
	Notice that the only way to obtain the entailment $\calK' \models_\brave D(a)$ is via the TBox axiom $C\sqcap A\subsum D$, since no axiom in \calK contains $D$.
	Therefore, we must have $\calK \models_\brave A(a)$.
	Moreover, we have $\calK \not\models_\brave B(a)$:
	Suppose to the contrary that $\calK \models_\brave B(a)$ and let $\calR$ be a witnessing repair such that $\tup{\calT, \calR} \models B(a)$.
	Then, we have that $\tup{\calT', \calR \cup \calH} \models \bot$, in particular due to the axiom $C \sqcap B \subsum \bot$ and the assertion $C(a)$.
  Since \calR is a repair, and hence \calT-consistent, there must be a conflict of $\tup{\calT', \calR \cup \Hyp}$ that is not a conflict of $\tup{\calT',\calA}$.
	But this leads to a contradiction to our assumption that \Hyp is conflict-confining for $\calK'$.
	As a result, $\calK'\not\models_\brave B(a)$ must be true.
	
	($\Leftarrow$) Suppose $\calK \models_\brave A(a)$ and $\calK \not\models_\brave B(a)$.
	Then, $\calK' \models_\brave A(a)$ and hence
  \[\tup{\calT', \calA\cup\Hyp} \models_\brave D(a).\]
	Therefore \Hyp is indeed a \brave-hypothesis for $\tup{\calK', \alpha}$.
	To show that \Hyp is conflict-confining for $\calK'$, suppose to the contrary that there is a conflict $\calC\in \Conf(\tup{\calT', \calA\cup\Hyp})$ such that $\calC\not\in \Conf(\tup{\calT', \calA})$.
	In particular, this implies that $C(a)\in\calC$ since $\calH=\{C(a)\}$.
	As a result, we have $\calC\setminus\{C(a)\}\subseteq \calA$ and $\calC$ is $\calT'$-consistent, and hence also \calT-consistent.
  But this implies that $\tup{\calT', \calC} \models B(a)$, since the only conflict involving $C(a)$ is via the axiom $C\sqcap B\subsum \bot$.
  Further, this means that $\tup{\calT, \calC} \models B(a)$, as the new axioms in $\calT'$ do not help entailment of $B(a)$.
	As \calC is \calT-consistent, there exists a repair $\calR \supseteq \calC$ of \calK, and we have $\tup{\calT, \calR} \models B(a)$.
  Consequently, $\calK \models_\brave B(a)$.
	But this is a contradiction to our assumption, so \Hyp must be conflict-confining.
\end{claimproof}
	We conclude by observing that the above reduction can be achieved in polynomial time. 
\end{proof}

\subsection{Proof Details and Constructions Used for \Cref{cor:verification-nt-min-el}}

\ClaimVerSubsetNT*
\begin{claimproof}
	($\Rightarrow$)
	Suppose $\Hyp$ is a $\subseteq$-minimal hypothesis for $Q(a)$ in $\calK'$.
	Observe that the only way to obtain the entailment $\calK' \models_\brave Q(a)$ is via the TBox axiom $C \sqcap D \sqcap A \subsum Q$, since no axiom in $\calK$ contains $Q$ and $Q(a) \not\in \Hyp$.
	Therefore, $\calK \models_\brave A(a)$, since otherwise $\calK' \not\models_\brave A(a)$, which would imply $\calK' \not\models_\brave Q(a)$.
	Moreover, we have $\calK \not\models_\brave B(a)$:
	Suppose to the contrary that $\calK \models_\brave B(a)$.
	Then $\calK' \models C(a)$ due to the axiom $B \sqsubseteq C$.
	Consequently, $\{D(a)\}$ is a \brave-hypothesis for $\alpha$ in $\calK'$, which is a contradiction to $\subseteq$-minimality of $\Hyp$.
	
	($\Leftarrow$) Suppose $\calK \models_\brave A(a)$ and $\calK \not\models_\brave B(a)$.
	Then, we also have $\calK' \models_\brave A(a)$ and hence
	\[\tup{\calT', \calA\cup\Hyp} \models_\brave Q(a),\]
	so $\Hyp$ is a $\brave$-hypothesis for $\tup{\calK', \alpha}$.
	For $\subseteq$-minimality, suppose to the contrary that there is a \brave-hypothesis $\Hyp' \subsetneq \Hyp$ for $\tup{\calK', \alpha}$.
	Since $D(a)$ cannot be entailed via any axiom in $\calK'$, we have $\Hyp' = \{D(a)\}$.
	However, this implies that $\calK' \models_\brave C(a)$, which can only be true if $\calK' \models_\brave B(a)$.
	But then $\calK \models_\brave B(a)$, which is a contradiction.
\end{claimproof}

The following example shows that the set of \ar-hypotheses for a given \ar-abduction problem does not need to be convex.
This underlies the hardness proof for verification of $\subseteq$-minimal \ar-hypotheses presented after.
\begin{example}\label{ex:ar-non-convex}
  Define
  \begin{align*}
  	\calT &\coloneqq \{ B_1 \sqcap B_2 \sqsubseteq \bot, \quad C_1 \sqcap C_2 \sqsubseteq \bot, \qquad B_1 \sqsubseteq A, \quad B_3 \sqsubseteq A\}, \\
    \calA &\coloneqq \{C_1(a), C_2(a)\}, \\
    \calB_1 &\coloneqq \{B_1(a)\}, \qquad \calB_2 \coloneqq \calB_1 \cup \{B_2(a)\}, \qquad \calB_3 \coloneqq \calB_2 \cup \{B_3(a)\}.
  \end{align*}
  and let $\calK \coloneqq \tup{\calT, \calA}$.
  We have $\calB_1 \subseteq \calB_2 \subseteq \calB_3$, and $\calB_1$ and $\calB_3$ are \ar-hypotheses for $\tup{\calK, A(a)}$, while $\calB_2$ is not.
\end{example}

\begin{theorem}\label{thm:elbot-verif-ar-subset}
	Verification of $\subseteq$-minimal \ar-hypotheses is \PiP-hard.
\end{theorem}

\begin{proof}
  We reduce from checking whether a given $\Pi_2$-QBF is true.
  Let $\Phi = \forall X \exists Y \varphi(X,Y)$ be a $\Pi_2$-QBF, where $X = \{x_1, \dots, x_n\}$ and $Y = \{y_1, \dots, y_m\}$.
  Note that we can assume w.l.o.g.\ that $\neg \varphi$ is in DNF with set of terms $\{c_1, \dots, c_k\}$.
  We construct an \ar-abduction problem $\tup{\calK, A(a)}$ and \ar-hypothesis $\Hyp$ for it s.t.\ $\forall X \exists Y \varphi$ is true iff $\Hyp$ is a $\subseteq$-minimal \ar-hypothesis.
  Equivalently, $\exists X \forall Y \neg \varphi$ is true iff there is some subset $\Hyp' \subsetneq \Hyp$ s.t.\ $\Hyp'$ is an \ar-hypothesis for $\tup{\calK, A(a)}$.

  The proof idea is as follows.
  We construct \calK and \Hyp in such a way that that subsets $\Hyp' \subsetneq \Hyp$ encode assignments over $X$, and repairs of $\tup{\calT, \calA \cup \Hyp'}$ range over encodings of all assignments over $Y$.
  Entailment of $A(a)$ in $\tup{\calT, \calA \cup \Hyp'}$ is then equivalent to the corresponding assignment over $X \cup Y$ satisfying $\varphi$.
  Further, we use the idea from \cref{ex:ar-non-convex}:
  To ensure entailment of $A(a)$ in $\tup{\calT, \calA \cup \Hyp}$, we use an additional axiom in \calT that circumvents the construction encoding $\Phi$, but cannot be triggered in $\tup{\calT, \calA \cup \Hyp'}$ for any subset $\Hyp' \subsetneq \Hyp$.
  Another important component of the construction will be two disjoint concepts $B_1$ and $B_2$ that allow us to split the set of repairs in two parts:
  The repairs containing the assertion $B_1(a)$ will ensure that $\varphi$ is satisfied for all assignments over $Y$, while the repairs containing the assertion $B_2(a)$ will ensure that \Hyp contains a full assignment over $X$.

  We now provide the definition of \calK, argue that $\tup{\calK, A(a)}$ is an \ar-abduction problem, and define the ABox \Hyp.
  We then provide further intuition on the components of the construction, followed by the proof of correctness.
  Define $\calK = \tup{\calT, \calA}$, where $\calT \coloneqq \calT_0 \cup \calT_1 \cup \calT_2$, and $\calT_0$, $\calT_1$, $\calT_2$ and \calA are defined as follows:
  \[
    \calT_0 \coloneqq \{\, B_1 \sqcap B_2 \sqsubseteq \bot, \quad B_1' \sqcap B_2  \sqsubseteq \bot \,\},
  \]
  \begin{align*}
    \calT_1 \coloneqq {} & \left\{ B_1 \sqcap \bigsqcap_{1 \leq i \leq n} (T_{x,i} \sqcap F_{x,i}) \sqsubseteq A \right\} \cup {} \\
    & \{\, B_1 \sqcap C_j \sqsubseteq A \mid 1 \leq j \leq k \,\} \cup {} \\
    & \{\, B_1' \sqcap T_{y,i} \sqcap F_{y,i} \sqsubseteq \bot \mid 1 \leq i \leq m \,\} \cup {} \\
    & \{\, B_1' \sqcap T_{x,i} \sqcap C_j \sqsubseteq \bot \mid \neg x_i \in c_j \,\} \cup {} \\
    & \{\, B_1' \sqcap F_{x,i} \sqcap C_j \sqsubseteq \bot \mid x_i \in c_j \,\} \cup {} \\
    & \{\, B_1' \sqcap T_{y,i} \sqcap C_j \sqsubseteq \bot \mid \neg y_i \in c_j \,\} \cup {} \\
    & \{\, B_1' \sqcap F_{y,i} \sqcap C_j \sqsubseteq \bot \mid y_i \in c_j \,\}
  \end{align*}
  \begin{align*}
    \calT_2 \coloneqq {} & \left\{ B_2 \sqcap \bigsqcap_{1 \leq i \leq n} \mathrm{HAVE}_i \sqsubseteq A \right\} \cup {} \\
    & \{\, T_{x,i} \sqsubseteq \mathrm{HAVE}_i,\; F_{x,i} \sqsubseteq \mathrm{HAVE}_i \mid 1 \leq i \leq n \,\}
  \end{align*}    
  \begin{align*}
    \calA \coloneqq {} & \{ B_1(a), B_1'(a), B_2(a) \} \cup {} \\
    & \{\, T_{y,i}(a), F_{y,i}(a) \mid 1 \leq i \leq m \,\} \cup {} \\
    & \{\, C_j(a) \mid 1 \leq j \leq k \,\}.
  \end{align*}
  First, note that $\tup{\calK, A(a)}$ is an \ar-abduction problem:
  the KB \calK is inconsistent, for example it has the conflict $\{B_1(a), B_2(a)\}$.
  Also, $\calK \not\models_\ar A(a)$ as there is a repair \calR of \calK with $B_2(a) \in \calR$.
  By the axioms in $\calT_0$, we have $B_1(a) \not\in \calR$, so $A(a)$ cannot be entailed by the aximos in $\calT_1$.
  But there is also no assertion of the form $T_{x,i}(a)$ or $F_{x,i}(a)$ in \calR, as these are not contained in \calA.
  Hence, $A(a)$ cannot be entailed by the first axiom in $\calT_2$, so $\tup{\calT, \calR} \not\models A(a)$.
  Now, define the \ar-hypothesis \Hyp for $\tup{\calK, A(a)}$ by
  \begin{align*}
    \Hyp \coloneqq \{\, T_{x,i}(a), F_{x,i}(a) \mid 1 \leq i \leq n \,\}.
  \end{align*}

  We now provide some intuition on the construction of \calK and \Hyp.
  The axioms in $\calT_0$ split the set of repairs into two parts:
  Those repairs that contain $B_1(a)$ and potentially $B_1'(a)$, and those repairs that contain $B_2(a)$.
  The axioms in $\calT_1$ only affect the former repairs while those in $\calT_2$ only affect the latter.
  Regarding the encoding of $\varphi$ in the construction, presence of an assertion of the form $T_{z,i}(a)$ or $F_{z,i}(a)$ encodes that variable $z$ is assigned to \texttt{true} or \texttt{false}, resp.
  Accordingly, repairs containing $B_1(a)$ encode assignments over $Y$, due to the axioms of the form $B_1' \sqcap T_{y,i} \sqcap F_{y,i} \sqsubseteq \bot$, while similarly subsets of \Hyp encode assignments over $Z$.
  (W.l.o.g., one can assume that \Hyp contains at most one of the assertions $T_{y,i}$ and $F_{y,i}$ for each $i$, due to how satisfaction of $\varphi$ is encoded.)
  Further, the assertions of the form $C_j(a)$ correspond to the terms $c_j$ of $\varphi$.

  Based on these ideas, $\calT_1$ ensures that for a given subset $\Hyp' \subsetneq \Hyp$, $A(a)$ is entailed in all repairs containing $B_1(a)$, iff the assignment $\theta_{\Hyp'}$ corresponding to $\Hyp'$ satisfies the formula $\forall Z \varphi$:
  The last $4$ kinds of axioms in $\calT_1$ ensure that $C_j(a)$ conflicts with assignments that falsify $c_j$, so some $C_j(a)$ only remains in all repairs (yielding entailment of $A(a)$), if at least one term of $\varphi[\theta_{\Hyp'}]$ is satisfied by each assignment over $Z$.
  This construction is circumvented for the full set \Hyp: Here, entailment of $A(a)$ is ensured via the first axiom in $\calT_1$.
  On the other hand, $\calT_2$ ensures that each \ar-hypothesis $\Hyp' \subsetneq \Hyp$ of $\tup{\calK, A(a)}$ contains at least one of the assertions $T_{y,i}(a)$, $F_{y,i}(a)$ (using the repairs containing $B_2(a)$).

  The full proof of correctness is now split into three parts, namely the proof that \Hyp is an \ar-hypothesis for $\tup{\calK, A(a)}$ and the two directions of the equivalence that there is an \ar-hypothesis $\Hyp' \subsetneq \Hyp$ of $\tup{\calK, A(a)}$ iff $\Phi$ is false, shown in Claims~\ref{cl:elbot-subsetmin-1}, \ref{cl:elbot-subsetmin-2} and~\ref{cl:elbot-subsetmin-3}.

  \begin{claim}\label{cl:elbot-subsetmin-1}
    \Hyp is an \ar-hypothesis for $\tup{\calK, A(a)}$.
  \end{claim}
  \begin{claimproof}
    Consider any repair $\calR$ of $\tup{\calT, \calA \cup \Hyp}$.
    We consider two cases.

    Case ``$B_1(a) \in \calR$''.
    By the axioms in $\calT_0$, we have $B_2(a) \not\in \calR$, so $A(a)$ cannot be entailed by the first axiom in $\calT_2$.
    If we have $C_j(a) \in \calR$ for some $j$, $A(a)$ is entailed by the axiom $B_1 \sqcap C_j \sqsubseteq A$.
    Otherwise, the only remaining disjointness axioms that could trigger are those of the form $B_1' \sqcap T_{y,i} \sqcap F_{y,i} \sqsubseteq \bot$, which do not affect the concepts of the form $T_{x,i}$ or $F_{x,i}$.
    Hence, by maximality of repairs, we have $T_{x,i}(a), F_{x,i}(a) \in \calR$ f.a.\ $1 \leq i \leq n$.
    Consequently, $A(a)$ is entailed by the first axiom in $\calT_1$.

    Case ``$B_1(a) \not\in \calR$''.
    First note that $B_1'(a) \not\in \calR$: Otherwise, we would have $B_2(a) \not\in \calR$ by the axioms in $\calT_0$.
    But then, $B_1(a) \in \calR$ by maximality of repairs, as $B_1$ does not occur in any disjointness axioms in $\calT_1$ or $\calT_2$.
    As $B_1'(a) \not\in \calR$, none of the disjointness axioms in $\calT_1$ can trigger, so we only need to consider those in $\calT_2$.
    Hence, we have $B_2(a), T_{x,i}(a), F_{x,i}(a) \in \calR$ for all $1 \leq i \leq n$ by maximality of repairs.
    Consequently, $A(a)$ is entailed via the axioms of the form $T_{x,i} \sqsubseteq \mathrm{HAVE}_i(a)$ and $F_{x,i} \sqsubseteq \mathrm{HAVE}_i(a)$ as well as the axiom $B_2 \sqcap \bigsqcap_{1 \leq i \leq n} \mathrm{HAVE}_i \sqsubseteq A$.
  \end{claimproof}

  \begin{claim}\label{cl:elbot-subsetmin-2}
    If $\forall X \exists Y \varphi(X,Y)$ is false, then there is a subset $\Hyp' \subsetneq \Hyp$ s.t.\ $\Hyp'$ is an \ar-hypothesis for $\tup{\calK, A(a)}$.
  \end{claim}
  \begin{claimproof}
    In this case, there is an assignment $s_X$ over $X$ s.t.f.a. assignments $s_Y$ over $Y$, there is conjunction term $c_j$ in $\neg \varphi$ with
    \[s_X \cup s_Y \models c_j.\]
    Define
    \[\Hyp' = \{\, T_{x,i}(a) \mid s_X(x_i) = 1 \,\} \cup \{\, F_{x,i}(a) \mid s_X(x_i) = 0 \,\}.\]
    Consider any repair $\calR$ of $\tup{\calT, \calA \cup \Hyp'}$.
    Similar to the above proof of claim, we distinguish two cases based on whether $B_1(a) \in \calR$.

    Case ``$B_1(a) \in \calR$''.
    By the axioms in $\calT_0$, we have $B_2(a) \not\in \calR$, so $A(a)$ cannot be entailed by the first axiom in $\calT_2$.
    If we have $B_1'(a) \not\in \calR$, then none of the disjointness axioms in $\calT_1$ can trigger, so there is some $C_j(a) \in \calR$ by maximality of repairs and we get the entailment of $A(a)$ via the axiom $B_1 \sqcap C_j \sqsubseteq A$.
    Otherwise, the disjointness axioms of the form $B_1' \sqcap T_{y,i} \sqcap F_{y,i}$ ensure that \calR contains at most one of the assertions $T_{y,i}(a)$ and $F_{y,i}(a)$ for each $1 \leq i \leq m$.
    Hence, the present assertions in \calR correspond to the assignment $s_\calR$ over $Y$ defined by
    \[s_\calR(y_i) \coloneqq \begin{cases}
      1, & \text{if } T_{y,i}(a) \in \calR \\
      0, & \text{if } F_{y,i}(a) \in \calR.
    \end{cases}\]
    By our assumption, there is some $c_j$ s.t. $s_X \cup s_\calR \models c_j$.
    By construction of $\calT_1$, this means that $C_j(a)$ is not in conflict with any of the assertions in $\Hyp'$ or any of the assertions $T_{y,i}(a)$ or $F_{y,i}(a)$ remaining in \calR.
    Hence, $A(a)$ is entailed by the axiom $B_1 \sqcap C_j \sqsubseteq A$.

    Case ``$B_1(a) \not\in \calR$''.
    First note that $B_1'(a) \not\in \calR$ by the same argument as in the previous proof of claim.
    As $B_1'(a) \not\in \calR$, none of the disjointness axioms in $\calT_1$ can trigger.
    Hence, we have $B_2(a) \in \calR$ by maximality of repairs.
    Further, by maximality of repairs and the fact that $s_X$ is an assignment over $X$, we have $T_{x,i}(a) \in \calR$ or $F_{x,i}(a) \in \calR$ for each $1 \leq i \leq n$.
    Consequently, $A(a)$ is entailed via the axioms of the form $T_{x,i} \sqsubseteq \mathrm{HAVE}_i$ and $F_{x,i} \sqsubseteq \mathrm{HAVE}_i$ as well as the axiom $B_2 \sqcap \bigsqcap_{1 \leq i \leq n} \mathrm{HAVE}_i \sqsubseteq A$.
  \end{claimproof}

  \begin{claim}\label{cl:elbot-subsetmin-3}
    If there is a subset $\Hyp' \subsetneq \Hyp$ s.t.\ $\Hyp'$ is an \ar-hypothesis for $\tup{\calK, A(a)}$, then $\forall X \exists Y \varphi(X,Y)$ is false.
  \end{claim}
  \begin{claimproof}
    By assumption we have $\tup{\calT, \calA \cup \Hyp'} \models_\ar A(a)$.
    First, note that $\Hyp' \cup \{ B_2(a) \}$ is \calT-consistent.
    Hence, there is a repair $\calR$ of $\tup{\calT, \calA \cup \Hyp'}$ with $\Hyp' \cup \{B_2(a)\} \subseteq \calR$.
    By the axioms in $\calT_0$, we have $B_1(a) \not\in \calR$.
    Hence, $A(a)$ cannot be entailed using axioms in $\calT_1$, so it must be entailed by the axiom $B_2 \sqcap \bigsqcap_{1 \leq i \leq n} \mathrm{HAVE}_i \sqsubseteq A$.
    This means that we have $T_{x,i}(a) \in \calR$ or $F_{x,i}(a) \in \calR$ for each $1 \leq i \leq n$.
    As these assertions do not occur in \calA, they must be contained in $\Hyp'$.

    Define the assignment $s_{\Hyp'}$ over $X$ by
    \[s_{\Hyp'}(x_i) \coloneqq \begin{cases}
      1, & \text{if } T_{x,i}(a) \in \Hyp' \\
      0, & \text{otherwise}
    \end{cases}\]
    and consider any assignment $s_Y$ over $Y$.
    As $s_Y$ is a function, the set
    \begin{align*}
      \calA_{s_Y} \coloneqq {} & \{B_1(a), B_1'(a)\} \cup \Hyp' \cup {} \\
      & \{\, T_{y,i}(a) \mid s_Y(y_i) = 1 \,\} \cup \{\, F_{y,i}(a) \mid s_Y(y_i) = 0 \,\}
    \end{align*}
    is $\calT$-consistent, so there are repairs containing all assertions in $\calA_{s_Y}$.
    In these repairs, $A(a)$ cannot be entailed by the axioms in $\calT_2$, since $B_2(a)$ is in conflict with $B_1(a)$.
    Further, it cannot be entailed by the first axiom in $\calT_1$, as $\Hyp' \subsetneq \Hyp$.
    Hence, there is some index $j$ s.t.\ $C_j(a)$ is not in conflict with any of the assertions in $\calA_{s_Y}$.
    (Otherwise, there is a repair $\calR$ among these that does not contain any of the assertions of the form $C_j(a)$, meaning that $\tup{\calT, \calR} \not\models A(a)$.)
    In particular, this is also true for the subset $\calA_{s_Y}' \subseteq \calA_{s_Y}$ defined by
    \[\calA_{s_Y}' \coloneqq \calA_{s_Y} \setminus \{\, F_{x,i}(a) \mid 1 \leq i \leq n, T_{x,i}(a) \in \calA_{s_Y}' \,\}.\]
    But the assertions of the forms $T_{x,i}(a)$, $F_{x,i}(a)$, $T_{y,i}(a)$ and $F_{y,i}(a)$ present in $\calA_{s_Y}'$ correspond to the assignment $s_{\Hyp'} \cup s_Y$, so this means that $c_j$ is not falsified by $s_{\Hyp'} \cup s_Y$ by construction of $\calT_1$. 
    Consequently, $s_{\Hyp'} \cup s_Y \models \neg \varphi$.
    As $s_Y$ was picked arbitrarily, this means that $\exists X \forall Y \neg \varphi(X,Y)$ is true and $\forall X \exists Y \varphi(X,Y)$ is false.
  \end{claimproof}
  Combined, the three above claims show that $\forall X \exists Y \varphi(X,Y)$ is true iff \Hyp is a $\subseteq$-minimal \ar-hypothesis for $\tup{\calK, A(a)}$, finishing the hardness proof.  
\end{proof}

\subsection{Proof details for \Cref{thm:verification-subset-min-el}}
To see that \Cref{claim:nt-subset-verification} still applies when we replace non-triviality by conflict-confinement, observe that $\calH\dfn \{C(a), D(a)\}$ is indeed conflict-confining since no axiom in $\calT'$ causes a new conflict involving either of the two (or both) assertions.

\subsection{Proof details for \Cref{cor:verification-cc-min-el} ($\PiP$-hardness for Brave semantics)}

We reuse the translation from $\forall\exists$-QBFs as in \Cref{sec:cc} and define $\calK\dfn \tup{\calT,\calA}$, where:
\begin{align*}
	\calT \coloneqq {} & \{\, C_d\sqcap \bigsqcap_{y\in Y} V_y \sqcap A_{\bar \varphi} \subsum C\, \} \cup {} \\
	%
	& \{\,C_d\sqcap A_y \subsum V_y, C_d\sqcap A_{\bar y} \subsum V_y \mid y\in Y\,\} \cup {} \\
	& \{\,A_x \sqcap A_{\bar x} \subsum \bot \mid x\in X\, \}  \cup {} \\
	%
	& \{\, C_d\sqcap A_\ell \subsum A_c \mid \ell \in c, c\in\varphi\, \}\cup {} \\ 
	& \{\,C_d\sqcap \bigsqcap_{y\in Y} V_y\sqcap \bigsqcap_{c\in \varphi}A_c\subsum A_\varphi\, \} \cup {} \\ 
	& \{\, A_\varphi \sqcap A_{\bar\varphi}\subsum \bot\,\} \cup \{\, A_\varphi\sqcap B_d \subsum \bot\, \}  \cup {} \\ 
	%
	& \{\, C_d\sqcap B_d \subsum \bot\, \} \cup {} 
	\{\, C\sqcap B_d \subsum \bot\, \} \cup {}\\ 
	& \{\, A_c\sqcap B_d \subsum \bot \mid c\in\varphi\, \} \cup {} 
	\{\, V_y\sqcap B_d \subsum \bot\, \} \text{ and} \\
	\calA \coloneqq & \{\, A_z(m), A_{\bar z}(m) \mid z\in Z \,\} \cup \{A_{\bar\varphi}(m), B_d(m), C_d(m)\},
\end{align*}

Let $N = |Y|$.
We now introduce fresh concept names $\{X_i \mid i\leq N+1\}$ and add additional axioms to obtain a new KB $\calK'\dfn \tup{\calT',\calA}$ with
$\calT' \dfn \calT\cup\{C_d\sqcap \bigsqcap_{1 \leq i \leq N+1} X_i \subsum C\}$.
Finally, let $\alpha\dfn C(m)$ as before and take  $\calH\dfn \{X_i(m)\mid i\leq N+1\}$. 
For correctness, we observe that \Hyp is a conflict-confining \brave-hypothesis for $C(m)$ of size $N+1$, whereas any conflict-confining \brave-hypothesis for $\tup{\calK, \alpha}$ has size $\leq N$ without loss of generality (since $N = |Y|$.
This implies the following equivalences:
\Hyp is not a $\leq$-minimal conflict-confining \brave-hypothesis for $\tup{\calK', C(m)}$ iff there is a conflict-confining \brave-hypothesis for $\tup{\calK', \alpha}$ of size $\leq N$ iff there is a conflict-confining \brave-hypothesis for $\tup{\calK, \alpha}$.

\begin{claim}\label{claim:el-cc-cmin-brave}
	$\Phi$ is false iff $\calH$ is not a $\leq$-minimal conflict-confining \brave-hypothesis for $\tup{\calK', \alpha}$.
\end{claim}
\begin{claimproof}
	The claim follows by observing that
	$\Phi$ is false iff $\tup{\calK, \alpha}$ admits a conflict-confining \brave-hypothesis (\Cref{claim:el-cc-brave})
	iff $\tup{\calK', \alpha}$ admits a conflict-confining \brave-hypothesis
	iff 
	$\calH$ is not a conflict-confining cardinality-minimal \brave hypothesis for $\tup{\calK', \alpha}$.
	Observe that $X_i$ for $i\leq N+1$ are fresh names not in $\calK$, hence $\calH$ is conflict-confining. 
	In other words, $\calH$ alone can not trigger any conflict in $\calK'$. 
	Moreover, even though entailing $C(m)$ in $\calK'$ triggers a conflict due to $C\sqcap B_d\subsum \bot$, this is not a new conflict as $\{C_d(m),B_d(m)\}$ is already a conflict in $\calK$ (and hence also in $\calK'$).
\end{claimproof}

\subsection{Proof details for \Cref{thm:verification-nt-min-conf-el}}

\begin{theorem}
   Verification of $\subseteq_c$-minimal non-trivial \brave-hypotheses is $\PiP$-hard.
\end{theorem}

\begin{proof}
	We reuse the construction from the hardness proof in \Cref{sec:cc}.
	Let $\Phi$ be a $\forall\exists$-QBF, and $\calK = \tup{\calT,\calA}$ be the KB obtained $\Phi$ using that construction.
	We add new axioms and assertions to \calK to obtain the KB $\calK' \dfn \tup{\calT,\calA'}$ as follows:
	\begin{align*}
		\calT' &\dfn \calT \cup \{C_d \sqcap X \subsum C, X\sqcap Y\subsum \bot\}, \text{ and} \\
		\calA' &\dfn \calA\cup \{Y(m)\}
	\end{align*}
	Then, we let $\alpha\dfn C(m)$ as before and take $\calH\dfn\{X(m)\}$.
	Here, $\calH$ induces exactly one more conflict in $\calK'$, namely $\{X(m), Y(m)\}$.
	
	For correctness, it is easy to see that there is a conflict-confining \brave-hypothesis for $\tup{\calK, C(m)}$ iff there is such a hypothesis for $\tup{\calK', C(m)}$.
	Furthermore, the letter is equivalent to \Hyp not being a $\subseteq_c$-minimal hypothesis for $\tup{\calK', C(m)}$, sincce existence \Hyp introduces exactly one new conflict for $\calK'$, while a conflict-confining \brave-hypothesis introduces no new conflicts.
	Thus the correctness follows due to the proof of Claim~\ref{claim:el-cc-brave}.
	This yields the mentioned $\PiP$-hardness.
\end{proof}

\end{document}
